\documentclass[12pt]{article}

\usepackage[english]{babel}
\usepackage[utf8]{inputenc}
\usepackage{amsmath}
\usepackage{graphicx}
\usepackage[colorinlistoftodos]{todonotes}
\usepackage{amssymb}
\usepackage{mathtools}
\usepackage{amsthm}
\usepackage{natbib}
\usepackage{tikz}
\usepackage{bbm}
\usepackage{adjustbox}
\usepackage{etoolbox}
\usepackage{enumerate}
\usepackage{pdflscape}
\usepackage{adjustbox}
\usepackage{scrextend}
\usepackage{titling}
\usepackage{setspace}
\usepackage{subcaption}
\usepackage{float}
\usepackage{bm}
\usepackage{mathrsfs}
\usepackage[margin=1.25in]{geometry}
\usepackage[normalem]{ulem}
\usepackage{hyperref}
\usepackage{mathtools}
\usepackage{arydshln}
\usepackage{comment}
\usepackage{pdflscape}
\usepackage{lscape}
\usepackage{afterpage}
\usepackage{multirow}

\hypersetup{pdfborder = 0 0 0,%
    pdftitle = A Bootstrap Test for the Existence of Moments for GARCH processes,%
    pdfauthor = Heinemann,%
}

\numberwithin{equation}{section}

\newcommand{\EE}{\mathbb{E}}
\newcommand{\PP}{\mathbb{P}}
\newcommand{\Cov}{\mathbb{C}\mbox{ov}}
\newcommand{\Var}{\mathbb{V}\mbox{ar}}

\newcommand{\Z}{\mathbb{Z}}
\newcommand{\R}{\mathbb{R}}
\newcommand{\N}{\mathbb{N}}

\newcommand{\iid}{\text{i.i.d.}}

\theoremstyle{definition}
\newtheorem{assumption}{Assumption}
\newtheorem{assumption*}{Assumption}

\newtheorem{algorithm}{Algorithm}
\newtheorem{example}{Example}

\theoremstyle{plain}
\newtheorem{theorem}{Theorem}
\newtheorem{lemma}{Lemma}
\newtheorem{corollary}{Corollary}
\newtheorem{proposition}{Proposition}

\theoremstyle{remark}
\newtheorem{remark}{Remark}

\newenvironment{Example}[1][Example]{\begin{trivlist}
\item[\hskip \labelsep {\bfseries #1}]}{\end{trivlist}}

\newcommand\blfootnote[1]{%
  \begingroup
  \renewcommand\thefootnote{}\footnote{#1}%
  \addtocounter{footnote}{-1}%
  \endgroup
}

\title{A Bootstrap Test for the Existence of Moments for GARCH Processes}

\author{Alexander Heinemann}

\date{\today}

\begin{document}

\begin{titlepage}
\begin{center}

\scshape
\LARGE{\textbf{A Bootstrap Test for the Existence of Moments for GARCH Processes}}
\normalfont
\end{center}
\begin{center}
\par\vspace{1cm}
\text{Alexander Heinemann$^\dagger$}
% \hspace{3cm}
% \text{Sean Telg$^\dagger$}
\par\vspace{2cm} %0.5
%\text{Department of Quantitative Economics}
\par
%\text{Maastricht University}
\par
\text{\today}
\vspace{1cm}
\end{center}

\begin{abstract}
This paper studies the joint inference on conditional volatility parameters and the innovation moments by means of bootstrap to test for the existence of moments for GARCH($p$,$q$) processes. We propose a residual bootstrap to mimic the joint distribution of the quasi-maximum likelihood estimators and the empirical moments of the residuals and also prove its validity. A bootstrap-based test for the existence of moments is proposed, which provides asymptotically correctly-sized tests without losing its consistency property. 
It is simple to implement and extends to other GARCH-type settings. A simulation study demonstrates the test's size and power properties in finite samples and an empirical application illustrates the testing approach.
\end{abstract}

\blfootnote{\hspace{-0.6cm}
%
%$^{\dagger}$e-mail: \href{mailto:a.heinemann@maastrichtuniversity.nl}{a.heinemann@maastrichtuniversity.nl}
%
$^{\dagger}$Department of Quantitative Economics, Maastricht University, Tongersestraat 53, 6211 LM Maastricht, Netherlands. E-mail address: \href{mailto:a.heinemann@maastrichtuniversity.nl}{a.heinemann@maastrichtuniversity.nl} 

}

\end{titlepage}

\newpage

\doublespacing

\section{Introduction}
\label{sec:7.1}

The existence of moments is key to statistical inference in financial time series.
While researchers generally assume that returns are strictly stationary, there is a large dispute to which extend their corresponding moments are finite. In particular many econometricians question the existence of fourth-order moments of returns, whereas some even challenge the existence of second-order moments. In the absence of moments many statistical tools become unreliable such as the ordinary least squares (OLS) estimator, whose asymptotic distribution requires the existence of fourth-order moments.
%
%; for example, the asymptotic distribution of the ordinary least squares (OLS) estimator requires the existence of fourth-order moments. % (c.f.\ \citeauthor{francq2011garch}, \citeyear{francq2011garch}, p.~129).
%
Frequently, returns are modeled as a product of a conditional volatility process and an sequence of innovations. In such case the existence of moments reduces to an inferential problem depending on the parameters of the conditional volatility model  and on characteristics of the innovation process. 
%
%For the well-known GARCH model (\citeauthor{engle1982autoregressive}, \citeyear{engle1982autoregressive}; \citeauthor{bollerslev1986generalized}, \citeyear{bollerslev1986generalized})
%
%Arguably the most popular models for log returns are autoregressive conditional heteroskedasticity (GARCH) models by \cite{engle1982autoregressive} and \cite{bollerslev1986generalized}. 
% 
\cite{ling1999probabilistic} and \cite{ling2002necessary} provide the necessary and sufficient condition for the existence of even-order moments in the well-known GARCH model. Similar results for other GARCH-type models are obtained by \cite{he1999properties}, \cite{ling2002stationarity} and \citeauthor{francq2011garch} (\citeyear{francq2011garch}, Chapter 10). Recently, \cite{francq2018testing} study the existence of moments for GARCH($1$,$1$) processes and derive the asymptotic distribution of the Wald statistic. Observing that the finite sample behavior is not always in par with the asymptotic results, they propose a bootstrap procedure, whose validity they prove for testing second-order stationarity. Unfortunately, neither for higher-order moments nor for higher-order GARCH models results are available. In particular the latter is a non-standard testing problem as the test-statistic is typically based on the spectral radius. In contrast, bootstrap methods are well-studied in conjunction with GARCH-type models \citep{hall2003inference,hidalgo2007goodness,corradi2008bootstrap,shimizu2009bootstrapping,cavaliere2018fixed,beutner2018residual,heinemann2018expected} and are also proven to be suitable in non-standard testing problems \citep{cavaliere2018bootstrap}.
%
% For GARCH-type models various bootstrap methods are studied: the sub-sample bootstrap \citep{hall2003inference}, the block bootstrap \citep{corradi2008bootstrap}, the wild bootstrap \citep{shimizu2009bootstrapping}, the residual bootstrap \citep{pascual2006bootstrap,hidalgo2007goodness,shimizu2009bootstrapping,cavaliere2018fixed,beutner2018residual,heinemann2018expected} and the smoothed bootstrap \citep{heinemann2018smoothed}. Bootstrap techniques do not only offer a powerful alternative to asymptotic theory, but are also proven to be suitable in non-standard testing problems \citep{cavaliere2018bootstrap}.
%
Therefore this paper studies the joint inference on conditional volatility parameters and the innovation moments by means of bootstrap. In particular, we prove the validity of the fixed design-residual bootstrap for a general class of volatility models and propose a bootstrap-based test for the existence of moments in the GARCH($p,q$) model. The testing procedure is simple to implement, provides asymptotically correctly-sized tests (without losing the consistency property) and can easily be extended to other GARCH-type settings.

The remainder of the paper is organized as follows. Section \ref{sec:7.2} describes the model. The joint asymptotic distribution of the quasi-maximum likelihood (QML) estimators of the volatility parameters and the empirical moments of the residuals is derived in Section \ref{sec:7.3}. In Section \ref{sec:7.4} we propose a fixed-design residual bootstrap method and prove its validity under mild assumptions. A bootstrap-based test for the existence of moments in the GARCH($p$,$q$) model is developed in Section \ref{sec:7.5} and extended to other GARCH-type models. A simulation study is conducted in Section \ref{sec:7.6} and an empirical application illustrates the bootstrap-based testing approach. Section \ref{sec:7.7} concludes. Proofs and auxiliary results are collected in the Appendix.

% \begin{itemize}

% 	\item No econometric tools are available for testing moments in higher order GARCH($p$,$q$) 

% 	\item As an application of our theory,
	
%     \item It is generally admitted 

% 	\item It is a well known fact that 

% 	\item Contribution is two-fold. We validate a refined bootstrap approach of \cite{francq2018testing} and extend it to higher
% \end{itemize}

\section{Model}
\label{sec:7.2}

We consider conditional volatility models of the form
\begin{align}
\label{eq:7.2.1}
\epsilon_t = \sigma_t\eta_t
\end{align}
with $t\in \Z$, where $\epsilon_t$ denotes the log-return, $\{\sigma_t\}$ is a volatility process and $\{\eta_t\}$ is a sequence of independent and identically distributed (\iid) variables. The volatility is assumed to be a measurable function of past observations
\begin{align}
\label{eq:7.2.2}
\sigma_{t}=\sigma_{t}(\theta_0)=\sigma(\epsilon_{t-1}, \epsilon_{t-2},\dots;\theta_0),
\end{align}
with $\sigma:\R^\infty\times \Theta\to(0,\infty)$ and $\theta_0$ denotes the true parameter vector belonging to the parameter space $\Theta \subset \R^r$, $r \in \N$. Various commonly used volatility models satisfy \eqref{eq:7.2.1}--\eqref{eq:7.2.2} such as GARCH($p$,$q$); for further examples see  \citeauthor{francq2015risk} (\citeyear{francq2015risk}, Table 1). Frequently we are not only interested in the parameter vector $\theta_0$, but also in characteristics of the innovation distribution.
%
%, say $\mu =(\mu_1,\dots, \mu_d)'= \EE[h(\eta_t)]$ with $h(x)=(x,x^2,\dots,x^d)'$ and $d \in \N$. 
%Later, we consider changes in the function $h$ to account for asymmetric GARCH models, etc.\\ \\
%
The following example illustrates.
\begin{example}
\label{ex:7.1}
Suppose $\{\epsilon_t\}$ follows a GARCH$(1,1)$ process given by \eqref{eq:7.2.1} and $\sigma_{t}^2 = \omega_0 + \alpha_0 \epsilon_{t-1}^2+ \beta_0 \sigma_{t-1}^2$, where $\theta_0 = (\omega_0,\alpha_0,\beta_0)'\in (0,\infty)\times[0,\infty)\times[0,1)$. Writing $\mu_{k}= \EE[\eta_t^{k}]$ for $k \in \N$, the necessary and sufficient condition for the existence of the fourth moment  $\EE[\epsilon_t^4]$ is 
\begin{align}
\label{eq:767698534}
\beta_0^2+2\alpha_0\beta_0\mu_2 +\alpha_0^2\mu_4<1.
\end{align}
\cite{francq2018testing} propose a Wald statistic based on  QML to test for \eqref{eq:767698534}.
\end{example}
%
%\begin{example}
%\label{ex:7.2}
%Suppose $\{\epsilon_t\}$ follows a T-GARCH$(1,1)$ process given by \eqref{eq:7.2.1} and $\sigma_{t+1} = \omega_0 + \alpha_0^+ \epsilon_t^+ + \alpha_0^- \epsilon_t^- + \beta_0 \sigma_{t}$, where $\theta_0 = (\omega_0,\alpha_0^+,\alpha_0^-,\beta_0)'\in (0,\infty)\times[0,\infty)^2\times[0,1)$. The necessary and sufficient condition for the existence of $\EE[|\epsilon_t|^d]$ is 
%
%\begin{align}
%\label{eq:87879759975}
%\sum_{k=0}^d \binom{d}{k} \Big(\big(\alpha_0^+\big)^k \mu_k^+ + \big(\alpha_0^-\big)^k \mu_k^-\Big) \beta_{0}^{m-k} \mu_{2k}<1
%\end{align}
%
%\end{example}
%
We collect the moment characteristics of the innovation distribution in a vector $\mu = \EE[h(\eta_t)]$, where 
%$h$ is some model-specific function with domain $\R$. For simplicity,
we confine ourselves here 
%--unless stated otherwise--
to the even moments, i.e.\  
\begin{align}
\label{eq:2187476}
%\mu = \EE[h(\eta_t)] \qquad \text{with} \qquad h(x)=(x,x^2,\dots,x^d)'
h(x)=(x^2,\dots,x^{2m})'
\end{align}
for some $m \in \N$. 
%\footnote{For alternative functional forms, we refer to Table \ref{tab:7.1}.}
%, corresponding to the GARCH model.
% For alternative functional forms, we refer to Table \ref{tab:7.1}. 
%for some function
%$\mu = \EE[h(\eta_t)]$, where 
%\begin{align}
%\mu = \EE[h(\eta_t)] \qquad \text{with} \qquad h(x)=(x,x^2,\dots,x^d)'
%\end{align}
%
%For the time being 
%while 
%and $d \in \N$. Here we confine ourselves to non-central moments for simplicity.\footnote{In Remark \ref{rem:7.3} we consider other functional forms of $h$.} 
Generally, $\mu$ is unknown
%to the researcher 
and needs to estimated just like $\theta_0$.

\section{Estimation}
\label{sec:7.3}

For the estimation of the parameters  $\theta_0$ and $\mu$ we use a two-step procedure, which is also employed by \citeauthor{francq2015risk} (2018). First, the vector of the conditional volatility parameters $\theta_0$ is estimated by QML. Since the conditional volatility $\sigma_{t}(\theta) = \sigma(\epsilon_{t-1},\dots,\epsilon_{1}, \epsilon_{0},\epsilon_{-1},\dots;\theta)$ can generally not be determined completely given a sample $\epsilon_1 ,\dots, \epsilon_n$, we replace the unknown presample observations by arbitrary values, say $\tilde{\epsilon}_t$, $t\leq 0$, yielding $\tilde{\sigma}_{t}(\theta) = \sigma(\epsilon_{t-1},\dots,\epsilon_{1}, \tilde{\epsilon}_{0},\tilde{\epsilon}_{-1},\dots;\theta)$. Then the QML estimator of $\theta_0$ is defined as a measurable solution $\hat{\theta}_n$ of 
\begin{align}
\label{eq:7.3.3}
\hat{\theta}_n=\arg\max_{\theta \in \Theta} \frac{1}{n}\sum_{t=1}^n \tilde{\ell}_t(\theta) \qquad \text{with} \qquad \tilde{\ell}_t(\theta)=-\frac{1}{2}\bigg(\frac{\epsilon_t}{\tilde{\sigma}_t(\theta)}\bigg)^2-\log \tilde{\sigma}_t(\theta).
\end{align}
In the second step, the first-step residuals are obtained, i.e. $\hat{\eta}_t=\epsilon_t/\tilde{\sigma}_t(\hat{\theta}_n)$, and the moments estimated:
\begin{align}
\label{eq:787979875}
\hat{\mu}_{n} = \frac{1}{n}\sum_{t=1}^n h(\hat{\eta}_t).
\end{align}
We first list several assumptions essential to the following analysis. Whereas in this paper we mainly focus on GARCH($p$,$q$) processes, the assumptions below are stated in a form that can readily applied to other GARCH-type processes (see Remark \ref{rem:7.3}).
\begin{assumption}{(Compactness)}
\label{as:7.1}
$\Theta$ is a compact subset of $\R^r$.
\end{assumption}

\begin{assumption}{(Stationarity \& Ergodicity)}
\label{as:7.2}
$\{\epsilon_t\}$ is a strictly stationary and ergodic solution of \eqref{eq:7.2.1} with \eqref{eq:7.2.2}.
\end{assumption}

\begin{assumption}{(Volatility process)}
\label{as:7.3}
For any real sequence $\{x_i\}$, the function $\theta\to\sigma(x_1,x_2,\dots;\theta)$ is continuous.  Almost surely, $\sigma_t(\theta)>\underline{\omega}$ for any $\theta \in \Theta$ and some $\underline{\omega}>0$ and $\EE[\sigma_t^s(\theta_0)]<\infty$ for some $s>0$. Moreover, for any $\theta \in \Theta$, we assume $\sigma_t(\theta_0)/\sigma_t(\theta)=1$ almost surely (a.s.) if and only if $\theta=\theta_0$.
\end{assumption}

\begin{assumption}{(Initial conditions)}
\label{as:7.4}
There exists a constant $\rho \in (0,1)$ and a random variable $C_1$ measurable with respect to $\mathcal{F}_0$ and $\EE[C_1^s]<\infty$ for some $s>0$ such that
\begin{enumerate}[(i)]
\item \label{as:7.4.1} $\sup_{\theta \in \Theta}|\sigma_t(\theta)-\tilde{\sigma}_t(\theta)|\leq C_1 \rho^t$;

\item \label{as:7.4.2} $\theta\to \sigma(x_1, x_2, \dots;\theta)$ has continuous second-order derivatives satisfying
\begin{align*}
\sup_{\theta \in \Theta}\bigg|\bigg|\frac{\partial \sigma_t(\theta)}{\partial \theta}-\frac{\partial \tilde{\sigma}_t(\theta)}{\partial \theta}\bigg|\bigg|\leq  C_1 \rho^t, \qquad \quad \sup_{\theta \in \Theta}\bigg|\bigg|\frac{\partial^2 \sigma_t(\theta)}{\partial \theta \partial \theta'}-\frac{\partial^2 \tilde{\sigma}_t(\theta)}{\partial \theta\partial \theta'}\bigg|\bigg|\leq C_1 \rho^t,
\end{align*}
where $||\cdot||$ denotes the Euclidean norm.
\end{enumerate}
\end{assumption}

\begin{assumption}{(Innovation process)} 
\label{as:7.5}
The innovations $\{\eta_t\}$ satisfy

\begin{enumerate}[(i)]
\item  \label{as:7.5.1}  $\eta_t\overset{iid}{\sim}F$ with $F$ being continuous, $\mu_2=1$, $\mu_4<\infty$ and $\eta_t$ is independent of $\{\epsilon_u:u<t\}$;

\item  \label{as:7.5.3}  $\EE\big[||h(\eta_t)||^2\big]<\infty$.% for some $d$ (to be specified);

%\item  \label{as:7.5.4} $M=\sup_{x\in \R}|x|f(x)<\infty$.
\end{enumerate}
\end{assumption}

\begin{assumption}{(Interior)}
 \label{as:7.6}
$\theta_0$ belongs to the interior of $\Theta$ denoted by $\mathring{\Theta}$.
\end{assumption}

\begin{assumption}{(Non-degeneracy)}
\label{as:7.7}
There does not exist a non-zero $\lambda\in \R^r$ such that $\lambda'\frac{\partial \sigma_t(\theta_0)}{\partial \theta}=0$ almost surely.
\end{assumption}

% \begin{assumption}{(Monotonicity)}
% \label{as:7.8}
% For any real sequence $\{x_i\}$ and for any $\theta_1,\theta_2 \in \Theta$ satisfying $\theta_1\leq \theta_2$ componentwise, we have $\sigma(x_1,x_2,\dots;\theta_1)\leq \sigma(x_1,x_2,\dots;\theta_2)$.
% \end{assumption}

\begin{assumption}{(Moments)}
\label{as:7.9}
There exists a neighborhood $\mathscr{V}(\theta_0)$ of $\theta_0$ such that the following variables have finite expectation:
\begin{align*}
\text{(i)}\sup_{\theta \in \mathscr{V}(\theta_0)}\bigg|\frac{ \sigma_t(\theta_0)}{\sigma_t(\theta)}\bigg|^a, \qquad \;\;\: \text{(ii)} \sup_{\theta \in \mathscr{V}(\theta_0)}\bigg|\bigg|\frac{1}{\sigma_t(\theta)}\frac{\partial \sigma_t(\theta)}{\partial \theta}\bigg|\bigg|^{b}, \qquad \;\;\: \text{(iii)} \sup_{\theta \in \mathscr{V}(\theta_0)}\bigg|\bigg|\frac{1}{\sigma_t(\theta)}\frac{\partial^2 \sigma_t(\theta)}{\partial \theta \partial \theta'}\bigg|\bigg|^c
\end{align*}
for some $a$, $b$, $c$ (to be specified).
\end{assumption}

\begin{assumption}{(Scaling Stability)}
\label{as:7.10}
 There exists a function $g$ such that for any $\theta \in \Theta$, for any $\lambda>0$, and any real sequence $\{x_i\}$
\begin{align*}
\lambda \sigma(x_1,x_2,\dots;\theta)=\sigma(x_1,x_2,\dots;\theta_\lambda),
\end{align*}
where $\theta_\lambda=g(\theta,\lambda)$ and $g$ is differentiable in $\lambda$.
\end{assumption}
%
%
% On the basis of the previous assumptions we state the strong consistency result of \citeauthor{francq2015risk} (\citeyear{francq2015risk}, Theorem 1) to the estimator of the innovation moments.
% %
% \begin{theorem}\textit{(Strong Consistency)}
% \label{thm:7.1} Under Assumptions \ref{as:7.1}--\ref{as:7.3}, \ref{as:7.4}(\ref{as:7.4.1}) and \ref{as:7.5}(\ref{as:7.5.1}) the estimator  in \eqref{eq:7.3.3} is strongly consistent, i.e. $\hat{\theta}_n \overset{a.s.}{\to} \theta_0$.
% %
% If in addition $\EE\big[|\eta_t|^{d}\big]<\infty$ and Assumption \ref{as:7.9}(i) holds with $a=d$, then the  estimator in \eqref{eq:787979875} satisfies $\hat{\mu}_{n} \overset{a.s.}{\to} \mu$.
% \end{theorem}
% %
% \begin{proof}
% \citeauthor{francq2015risk} (\citeyear{francq2015risk}, Theorem 1) establish $\hat{\theta}_n \overset{a.s.}{\to} \theta_0$. The second claim follows from \citeauthor{beutner2018residual} (\citeyear{beutner2018residual}, Lemma 2).
% \end{proof}
%
The assumptions are fairly standard in the literature; for a discussion we refer to \cite{francq2015risk} and \cite{beutner2018residual}.
To lighten notation, we henceforth write $D_t(\theta) =\frac{1}{\sigma_t(\theta)}\frac{\partial\sigma_t(\theta)}{\partial \theta}$  and drop the argument when evaluated at the true parameter, i.e.\ $D_t=D_t(\theta_0)$. In addition, we define $d=deg(h)$, the highest polynomial degree of the function $h$, 
%denote the polynomial degree of the function $h$ by $d$, i.e.\ $d=deg(h)$,
which reduces to $2m$ using \eqref{eq:2187476}. The next result provides the joint asymptotic distribution of $\hat{\theta}_n$ and $\hat{\mu}_{n}$. A similar result for a GARCH($p$,$q$) model can be found in \cite{francq2018testing}.
\begin{theorem}\textit{(Asymptotic distribution)}
\label{thm:7.2} Suppose Assumptions \ref{as:7.1}--\ref{as:7.10} hold with $a=\max\{4,2d\}$, $b=4$ and $c=2$. Then, we have
\begin{align}
\label{eq:7.3.6}
      \begin{pmatrix}
      \sqrt{n}(\hat{\theta}_n-\theta_0)\\
    \sqrt{n}(\hat{\mu}_{n} - \mu)
      \end{pmatrix}
\overset{d}{\to}N\big(0, \Sigma \big) \qquad \mbox{with}\qquad
\Sigma=
      \begin{pmatrix}
      \frac{\mu_4-1}{4}J^{-1} & -J^{-1}\Omega \nu'\\
    -\nu \Omega'J^{-1} & \Xi
      \end{pmatrix},
\end{align}
where $\Omega = \EE[D_t]$, $J=\EE[D_tD_t']$, $\nu = \EE\big[\eta_t\frac{\partial h(\eta_t)}{\partial x}\big]$, $ \Xi = \frac{\mu_4-1}{4}\nu \nu' +\frac{1}{2}(\xi \nu'+\nu \xi')+\Upsilon$, $\Upsilon = \Var[h(\eta_t)]$ and $\xi = \Cov[h(\eta_t),\eta_t^2]$.
\end{theorem}
The asymptotic distribution in Theorem \ref{thm:7.2} can be used to perform inference on parameters after having obtained a consistent estimator for $\Sigma$. A powerful alternative to perform statistical inference provide bootstrap methods.

\section{Bootstrap}
\label{sec:7.4}

We employ a fixed-design residual bootstrap scheme as in \cite{cavaliere2018fixed} and \cite{beutner2018residual} to approximate the distribution of the estimators in \eqref{eq:7.3.3}--\eqref{eq:787979875}. We indicate the bootstrap quantities by a  superscript $^*$ and use the usual bootstrap notation:  “$\overset{p^*}{\to}$", “$\overset{d^*}{\to}$", “$O_{p^*}(1)$", “$o_{p^*}(1)$", $\PP^*$ and $\EE^*$ (cf.\ \citeauthor{chang2003sieve}, \citeyear{chang2003sieve}).

%For the setup of the bootstrap algorithm we assume the availability of a consistent estimator of $\theta_0$ denoted by $\check{\theta}_n$. In many applications $\check{\theta}_n$ coincides with $\hat{\theta}_n$, yet in Section \ref{sec:7.5} we consider a constrained estimator version.
%
\begin{algorithm}\textit{(Fixed-design residual bootstrap)}
\label{alg:5.1}
\begin{enumerate}

\item For $t=1,\dots, n$, generate  $\eta_t^* \overset{iid}{\sim} \hat{\mathbbm{F}}_n$ and the bootstrap observation $\epsilon_t^* = \tilde{\sigma}_t(\hat{\theta}_n) \eta_t^*$.

\item Calculate the bootstrap estimator 
\begin{align}
\label{eq:7.4.1}
\hat{\theta}_n^* = \arg \max_{\theta \in \Theta}\frac{1}{n}\sum_{t=1}^n \ell_t^*(\theta) \quad \text{with} \quad \ell_t^*(\theta)=-\frac{1}{2}\bigg(\frac{\epsilon_t^{*}}{\tilde{\sigma}_t(\theta)}\bigg)^2-\log \tilde{\sigma}_t(\theta).
\end{align}
\item For $t=1,\dots,n$ compute the bootstrap residual $\hat{\eta}_t^* = \epsilon_t^*/\tilde{\sigma}_t(\hat{\theta}_n^*)$
and obtain 
\begin{align}
\label{eq:7.4.2}
\hat{\mu}_{n}^*: = \frac{1}{n}\sum_{t=1}^n h\big(\hat{\eta}_t^*\big).
\end{align}
\end{enumerate}
\end{algorithm}
\noindent
The asymptotic validity of the bootstrap procedure is stated in the following theorem.
\begin{theorem}\textit{(Boostrap consistency)}
\label{thm:7.3}
If Assumptions \ref{as:7.1}--\ref{as:7.10} hold with $a=- 12, \max\{2d,12\}$, $b=12$ and $c=6$, then
\begin{align}
\label{eq:76784132}
      \begin{pmatrix}
      \sqrt{n}(\hat{\theta}_n^*-\hat{\theta}_n)\\
    \sqrt{n}(\hat{\mu}_{n}^* - \hat{\mu}_{n})
      \end{pmatrix}
\overset{d^*}{\to} N(0, \Sigma),
\end{align}
%
%where $\overset{d}{\to}_p$ denotes convergence in distribution given the original sample 
in probability.
\end{theorem}
\begin{remark}
\label{rem:7.1}
The estimator $\hat{\theta}_n$ in the first step of Algorithm \ref{alg:5.1} can be replaced by any consistent estimator of $\theta_0$, say $\check{\theta}_n$. A close inspection of the proof of Theorem \ref{thm:7.3} reveals that the bootstrap's consistency follows after an appropriate standardization, i.e.\  replace $\hat{\theta}_n$ by $\check{\theta}_n$ in \eqref{eq:76784132}.
\end{remark}
In the subsequent section we employ Theorem  \ref{thm:7.3} and Remark \ref{rem:7.1} to derive a bootstrap-based test for the existence of moments in the GARCH model.

\section{Bootstrap Test for the Existence of Moments}
 \label{sec:7.5}

%We begin with the GARCH model and later indicate the relevant changes when turning to GARCH extensions.
%\subsection{GARCH($p$,$q$) Model}
We consider a GARCH($p$,$q$) model, in which the recursive form of \eqref{eq:7.2.2} is given by
\begin{align}
\label{eq:897750}
    \sigma_{t}^2 = \omega_0+ \sum_{i=1}^q\alpha_{0i} \epsilon_{t-i}^2+ \sum_{j=1}^p\beta_{0j} \sigma_{t-j}^2,
\end{align}
where $\theta_0 = (\omega_0,\alpha_{01},\dots,\alpha_{0q},\beta_{01},\dots, \beta_{0p})'\in \R_{>0} \times \R_{\geq 0}^{p+q}$. 
%satisfies $\omega_0>0$, $\alpha_{0i}\geq 0$ for all $i$ and $\beta_{0j}\geq 0$ for all $j$.
We are interested in testing whether for this GARCH process the moment $\EE[\epsilon_t^{2m}]$ exists.
%, where we focus on even moments, i.e.\ $d=2m$ for some $m \in \N$.
\cite{ling1999probabilistic} and \cite{ling2002necessary} provide the necessary and sufficient condition for the existence of even moments of model \eqref{eq:7.2.1} and \eqref{eq:897750}. For any matrix $A$ we write $||A||_S$ to denote its spectral norm, i.e.\ 
 %largest singular value of $A$,
$||A||_S=\sqrt{\lambda_{\max} (A'A)}$, %where $\lambda_{\max} (A'A)$ is the maximum eigenvalue of $A'A$.
and set $A^{\otimes m} = A \otimes A \otimes \dots \otimes A$ ($m$ factors), where $\otimes$ is the Kronecker product. Then the moment $\EE[\epsilon_t^{2m}]$ of the GARCH process is finite if and only if $T=\big|\big|\EE[A_t^{\otimes m}]\big|\big|_S<1$,
%
% \begin{align}
% \big|\big|\EE[A_t^{\otimes m}]\big|\big|_S<1,
% \end{align}
%
where $A_t =A(\theta_0,\eta_t)$ and
\begin{align}
\label{eq:689954}
A(\theta,\eta) = \left(\begin{array}{@{}ccc|ccc@{}}
    \alpha_1 \eta^2 & \dots & \alpha_q \eta^2 & \beta_1 \eta^2 & \dots & \beta_p \eta^2 \\
     & I_{(q-1)\times(q-1)} & O_{(q-1)\times 1} &  & O_{(q-1)\times p} & \\\hline
    \alpha_1  & \dots & \alpha_q  & \beta_1  & \dots & \beta_p  \\
     & O_{(p-1)\times q} &  &  & I_{(p-1)\times(p-1)} & O_{(p-1)\times 1}
  \end{array}\right).
\end{align}
We are interested in testing the null hypothesis $H_0$: $T < 1$ against the alternative hypothesis $H_1$: $T\geq 1$. As usual in hypothesis testing where the null hypothesis is characterized by an open set, the test is in fact constructed for the closure of $H_0$, i.e.\ 
%
% \begin{align}
% \label{eq:998431298}
% H_0: T < 1 \qquad \quad \text{against} \qquad \quad H_1: T\geq 1.
% \end{align}
\begin{align}
\label{eq:998431298}
\bar{H}_0: T \leq 1 \qquad \quad \text{against} \qquad \quad \bar{H}_1: T > 1.
\end{align}
%
% Alternatively, one may be interested in testing 
% %
% \begin{align}
% \label{eq:547984}
% \tilde{H}_0:T\geq 1 \qquad \quad \text{against} \qquad \quad \tilde{H}_1:T< 1.
% \end{align}
%
%$H_0: T<1$ against the alternative $H_1: T\geq 1$.
%As usual in hypothesis testing where the null hypothesis is characterized by an open set, the test is in fact constructed for the closure of the $H_0$, i.e. $\bar{H}_0: T\leq 1$ against the alternative $\bar{H}_1: T> 1$.
%
% \begin{align}
% \bar{H}_0: T\leq 1 \qquad \quad \text{vs.} \qquad \quad \bar{H}_1: T> 1.
% \end{align}
%
Before proceeding with the test statistic, note that $T$ can be expressed in terms of $\theta_0$ and $\mu$. To illustrate this fact, we review the GARCH($1$,$1$) model from Example \ref{ex:7.1}.
\begin{Example}\textbf{1.} (\textit{continued})
We observe that the left-hand side of \eqref{eq:767698534}  corresponds to $T$ for $m=2$. Further, for $p=q=1$ we find that $A_t$ reduces to $A_t = (\eta_t^2, 1)'(\alpha_{01}, \beta_{01})$, such that for general $m$ we have $\EE[A_t^{\otimes m}]=\EE\big[(\eta_t^2, 1)^{\prime \otimes m}\big](\alpha_{01}, \beta_{01})^{\otimes m}$. The latter possesses a single non-zero eigenvalue (c.f.\ \citeauthor{francq2011garch}, \citeyear{francq2011garch}, p.\ 45) given by
\begin{align}
\label{eq:767447}
\big|\big|\EE[A_t^{\otimes m}]\big|\big|_S=\sum_{k=0}^m \binom{m}{k} \alpha_{01}^k \beta_{01}^{m-k} \mu_{2k}.
\end{align}
\end{Example}
To appreciate why $T$ is a function of $\theta_0$ and $\mu$
%this claim holds
also in higher order GARCH models, we state the following proposition. 
%
%Hence, in the GARCH($1$,$1$), $\big|\big|\EE[A_t^{\otimes m}]\big|\big|_S$ can be expressed as a function of $\theta$ and $\mu$ (in fact only the even moments matter here). The question arises whether this holds true in higher-order GARCH models. The answer is affirmative as the following Lemma indicates.
%
\begin{proposition}
\label{prop:7.1}
For all $m \in \N$, we have $A(\theta,\eta)^{\otimes m}= \sum_{k=0}^m B_{k,m}(\theta) \eta^{2k}$ with $A(\theta,\eta)$ given in \eqref{eq:689954} and
$\big\{B_{k,m}(\theta): k=0,1,\dots,m\big\}$ is a sequence of matrices, where each matrix has dimension $(p+q)^m\times(p+q)^m$ and depends on $\theta$.
\end{proposition}
Employing Proposition \ref{prop:7.1}, one finds $\EE\big[A_t^{\otimes m}\big]= \sum_{k=0}^m B_{k,m}(\theta_0) \mu_{2k}$ with $\mu_{0}=1$ and hence there exists a function $\tau:\Theta \times \R^{dim(\mu)} \to \R_+$ such that
\begin{align}
\label{eq:8787521}
T=\tau(\theta_0, \mu )= \big|\big|\EE[A(\theta_0,\eta_t)^{\otimes m}]\big|\big|_S.
\end{align}
With regard to Section \ref{sec:7.3}, a natural test statistic is given by
\begin{align}
\label{eq:983298321}
\hat{T}_n = \tau(\hat{\theta}_n,\hat{\mu}_n) = \bigg|\bigg|\frac{1}{n}\sum_{t=1}^n A(\hat{\theta}_n,\hat{\eta}_t)^{\otimes m}\bigg|\bigg|_S.
\end{align}
%
%To find critical values that control the size of the test, we study the ``worst case scenario'' under $\bar{H}_0$, which corresponds to $T$ being equal to $1$. With regard to \eqref{eq:767447}, one can rely on asymptotic theory when $p=q=1$.
For $p=q=1$ one can rely on asymptotic theory to find critical values that control the size of the test.
\begin{corollary}
\label{cor:7.1}
Suppose a GARCH(1,1) process $\{\epsilon_t\}$ with parameter $\theta_0$ and \iid~sequence $\{\eta_t\}$, which satisfies the assumptions of Theorem \ref{thm:7.2}. Then 
\begin{align}
\sqrt{n}(\hat{T}_n-T) \overset{d}{\to}N(0,\varsigma^2),
\end{align}
where $\varsigma^2 =  \frac{\partial \tau(\theta_0,\mu)}{\partial (\theta',\bar{\mu}')}\Sigma \frac{\partial \tau(\theta_0,\mu)}{\partial (\theta',\bar{\mu}')'}$.
\end{corollary}
The previous corollary is a direct consequence of Theorem \ref{thm:7.2} and the delta-method. Hence, testing $\bar{H}_0$: $T\leq 1$ in the GARCH($1$,$1$) at the asymptotic level $\alpha \in (0,1)$ could be defined by the rejection region $\big\{\sqrt{n}(\hat{T}_n-1) >\hat{\varsigma}_n\Phi^{-1}(1-\alpha)\big\}$, where $\hat{\varsigma}_n$ is a consistent estimate for $\varsigma$ and $\Phi$ denotes the standard normal cumulative distribution function.
%Hence, a test of $\EE[\epsilon_t^{2m}]<\infty$  at the asymptotic level $\alpha \in (0,1)$ could be defined by the rejection region $\big\{\sqrt{n}(\hat{T}_n-1)/\hat{\varsigma}_n>\Phi^{-1}(1-\alpha)\big\}$, where $\hat{\varsigma}_n$ is a consistent estimate for $\varsigma$ and $\Phi$ denotes the standard normal cumulative distribution function. 
However, as shown in \cite{francq2018testing}, the finite sample distribution of $\hat{T}_n$ is not always in par with the asymptotic results. Moreover, for higher order GARCH models, this asymptotic approach is practically infeasible due to the complicated form of function $\tau$ (recall that $\tau$ is a composite function involving the spectral norm). Instead we propose to mimic the finite sample distribution of the test statistic by means of a bootstrap procedure similar to Section \ref{sec:7.4}. To construct such bootstrap scheme we re-estimate the parameter $\theta$ to impose the null hypothesis $\bar{H}_0$ for the ``bootstrap world''. We denote the constrained estimator by $\hat{\theta}_n^c$, which satisfies
\begin{align}
\label{eq:238748761}
\hat{\theta}_n^c=\arg\max_{\theta \in \Theta_n^c} \frac{1}{n}\sum_{t=1}^n \tilde{\ell}_t(\theta)  \qquad \text{with} \qquad  \Theta_n^c = \big\{\theta \in \Theta : \tau(\theta,\hat{\mu}_n)\leq 1\big\}.
\end{align}
%
% with
% %
% \begin{align}
% \label{eq:238748761}
%   \Theta_n^c = \big\{\theta \in \Theta : \tau(\theta,\hat{\mu}_n)\leq 1\big\}. 
% \end{align}
%
This estimator is strongly consistent for $\theta_0$  when $\tau(\theta_0,\mu)\leq 1$; for details we refer to Lemma \ref{lem:7.1} in the Appendix. Note that, by construction, the corresponding constrained test statistic  
\begin{align}
\hat{T}_n^c = \tau(\hat{\theta}_n^c,\hat{\mu}_n) = \bigg|\bigg|\frac{1}{n}\sum_{t=1}^n A(\hat{\theta}_n^c,\hat{\eta}_t)^{\otimes m}\bigg|\bigg|_S
\end{align}
satisfies $\hat{T}_n^c\leq 1$. Based on the constrained estimator $\hat{\theta}_n^c$ we propose a fixed-design residual bootstrap algorithm to mimic the distribution of the test statistic $\hat{T}_n$. 

\begin{algorithm}\textit{(Fixed-design residual bootstrap)}
\label{alg:5.2}
\begin{enumerate}

\item For $t=1,\dots, n$, generate  $\eta_t^\star \overset{iid}{\sim} \hat{\mathbbm{F}}_n$ and the bootstrap observation $\epsilon_t^\star = \tilde{\sigma}_t(\hat{\theta}_n^c) \eta_t^\star$.

\item Calculate the bootstrap estimator 
\begin{align}
\label{eq:7.5.9}
\hat{\theta}_n^\star = \arg \max_{\theta \in \Theta}\frac{1}{n}\sum_{t=1}^n \ell_t^\star(\theta) \quad \text{with} \quad \ell_t^\star(\theta)=-\frac{1}{2}\bigg(\frac{\epsilon_t^{\star}}{\tilde{\sigma}_t(\theta)}\bigg)^2-\log \tilde{\sigma}_t(\theta).
\end{align}
\item For $t=1,\dots,n$ compute the bootstrap residual $\hat{\eta}_t^\star = \epsilon_t^\star/\tilde{\sigma}_t(\hat{\theta}_n^\star)$
and obtain the bootstrap test statistic
\begin{align}
\label{eq:7.5.10}
\hat{T}_n^\star = \tau(\hat{\theta}_n^\star,\hat{\mu}_n^\star) =  \bigg|\bigg|\frac{1}{n}\sum_{t=1}^n \big(A(\hat{\theta}_n^\star,\hat{\eta}_t^\star)\big)^{\otimes m}\bigg|\bigg|_S
\end{align}
with $\hat{\mu}_{n}^\star = \frac{1}{n}\sum_{t=1}^n h\big(\hat{\eta}_t^\star\big)$.
\end{enumerate}
\end{algorithm}
Since the bootstrap quantities are generated under the constrained estimator, a superscript $^\star$ is employed to distinguish them from the ones in Algorithm \ref{alg:5.1}. The corresponding bootstrap notation is given by:  “$\overset{p^\star}{\to}$", “$\overset{d^\star}{\to}$", “$O_{p^\star}(1)$", “$o_{p^\star}(1)$", $\PP^\star$ and $\EE^\star$. The bootstrap procedure described in  Algorithm \ref{alg:5.2} is valid in the following sense.
%
% \begin{corollary}
% \label{cor:7.2}
% Suppose the assumptions of Theorem \ref{thm:7.3} hold true. If $T=1$, we have
% %
% \begin{align}
% \label{eq:5850320}
% \sup_x\Big|\PP^\star\big[\sqrt{n}(\hat{T}_n^\star-1)\leq x\big]-\PP\big[\sqrt{n}(\hat{T}_n-1)\leq x\big]\Big|\overset{p}{\to}0.
% \end{align}
% Under the alternative, i.e.\ $T>1$, we have $\sqrt{n}(\hat{T}_n^\star-1)=O_{p^\star}(1)$ in probability.
% \end{corollary}
\begin{corollary}
\label{cor:7.2}
Suppose the assumptions of Theorem \ref{thm:7.3} hold true. Under the null hypothesis $\bar{H}_0$: $T \leq 1$ we have
\begin{align}
\label{eq:5850320}
\sup_x\Big|\PP^\star\big[\sqrt{n}(\hat{T}_n^\star-\hat{T}_n^c)\leq x\big]-\PP\big[\sqrt{n}(\hat{T}_n-T)\leq x\big]\Big|\overset{p}{\to}0.
\end{align}
Under the alternative $\bar{H}_1$: $T > 1$ we have $\sqrt{n}(\hat{T}_n^\star-\hat{T}_n^c)=O_{p^\star}(1)$ in probability.
\end{corollary}
%
%The proof is based on Theorem \ref{thm:7.3} (see also Remark \ref{rem:7.1}) and arguments given in \cite{watson1983statistics} and \cite{kato2013perturbation}.
The previous corollary legitimatizes the following bootstrap test to assess whether $\EE[\epsilon_t^{2m}]$ is finite in the GARCH($p$,$q$) model. We acquire a set of $B$ bootstrap replicates, i.e.\ $\hat{T}_n^{\star (b)}$ for $b=1,\dots, B$, by repeating Algorithm \ref{alg:5.2}  and compute 
\begin{align}
\label{eq:21872}
    \hat{p}_{n,B}^\star = \frac{1}{B}\sum_{b=1}^B\mathbbm{1}{\big\{\hat{T}_n -1\leq \hat{T}_n^{\star (b)}-\hat{T}_n^c\big\}},
\end{align}
 which proxies the p-value of the null hypothesis $\bar{H}_0: T\leq 1$.
Thus, one rejects the null hypothesis when \eqref{eq:21872} is below the nominal level of the test (e.g.\ $5\%$ or $10\%$). To appreciate why the bootstrap test is consistent, we note that under the alternative $\bar{H}_1$: $T > 1$ we have $\sqrt{n}(\hat{T}_n^\star-\hat{T}_n^c)=O_{p^\star}(1)$ whereas
\begin{align}
\sqrt{n}(\hat{T}_n-1) = \underbrace{\sqrt{n}(\hat{T}_n-T)}_{=O_p(1)}+ \underbrace{\sqrt{n}(T-1)}_{\to \infty}
\end{align}
diverges in probability.
%
% \begin{remark}
% For a given test level $\alpha \in (0,1)$ the bootstrap test is defined by the rejection region $\big\{\sqrt{n}(\hat{T}_n-1)>G_{n,B}^{\star\:-1}(1-\alpha)\big\}$, where  $G_{n,B}^{\star}(x)= \frac{1}{B}\sum_{b=1}^B\mathbbm{1}{\big\{\sqrt{n}(\hat{T}_n^{\star (b)}-\hat{T}_n^c)\leq x\big\}}$ approximates $G_{n}^{\star}(x) = \PP^\star\big[\sqrt{n}(\hat{T}_n^\star-\hat{T}_n^c)\leq x\big]$ for large $B$.
% \end{remark}

\begin{remark}
\label{rem:7.2}
In case one is interested in the null hypothesis $\tilde{H}_0:T\geq 1$ against the alternative hypothesis $\tilde{H}_1:T< 1$, the outlined bootstrap testing procedure can be readily adapted: replace ``$\leq$"  in equations \eqref{eq:238748761} and \eqref{eq:21872} by ``$\geq$".
\end{remark}

\afterpage{
\begin{landscape}

\begin{table}[]
\centering
\resizebox{21cm}{!}{
\begin{tabular}{ccccc}
          & \textbf{Moment}                       & \textbf{Volatility Recursion}                                                                                                                                                                                                                                            & $\bm{A(\theta,\eta)}$                                                                                                                                                                                                                                                                                                                                                                                                                                                                                                                                                                                                                                                                                                             & $\bm{h(x)}$                                                                                                                                                                                                     \\ \hline \hline
          &                              &                                                                                                                                                                                                                                                  &                                                                                                                                                                                                                                                                                                                                                                                                                                                                                                                                                                                                                                                                                                                              &                                                                                                                                                                                                            \\
ARCH      & $\EE[\epsilon_t^{2m}]$       & $\sigma_{t}^2 = \omega_0+ \sum\limits_{i=1}^q\alpha_{0i} \epsilon_{t-i}^2$                                                                                                                               
& \begin{tabular}[c]{@{}c@{}}$\left(\begin{array}{@{}ccc@{}}     \alpha_1 \eta^2 & \dots & \alpha_q \eta^2  \\      & I_{(q-1)\times(q-1)} & O_{(q-1)\times 1}\\   \end{array}\right)$\end{tabular}                                                                                                                                                                                                                                                                                                                                                                                                                                                                                                                        & \begin{tabular}[c]{@{}c@{}}$\left(\begin{array}{@{}c@{}}     x^2  \\ \vdots \\      x^{2m}\\   \end{array}\right)$\end{tabular}                                                                      \\
          &                              & \multicolumn{1}{l}{}                                                                                                                                                                                                                             & \multicolumn{1}{l}{}                                                                                                                                                                                                                                                                                                                                                                                                                                                                                                                                                                                                                                                                                                         & \multicolumn{1}{l}{}                                                                                                                                                                                       \\
GARCH  & $\EE[\epsilon_t^{2 m}]$ & \begin{tabular}[c]{@{}c@{}}$\sigma_{t}^2 = \omega_0 + \sum\limits_{i=1}^q\alpha_{0i}\epsilon_{t-i}^2$ \\ $\textcolor{white}{aaaaaaa} + \sum\limits_{j=1}^p\beta_{0j} \sigma_{t-j}^2 $\end{tabular}     & \begin{tabular}[c]{@{}c@{}}$\left(\begin{array}{@{}ccc|ccc@{}}     \alpha_1 \eta^2 & \dots & \alpha_q \eta^2 & \beta_1 \eta^2 & \dots & \beta_p \eta^2 \\      & I_{(q-1)\times(q-1)} & O_{(q-1)\times 1} &  & O_{(q-1)\times p} & \\\hline    \alpha_1  & \dots & \alpha_q  & \beta_1  & \dots & \beta_p  \\      & O_{(p-1)\times q} &  &  & I_{(p-1)\times(p-1)} & O_{(p-1)\times 1}\\   \end{array}\right)$\end{tabular}                                                                                                                                                                                                    & \begin{tabular}[c]{@{}c@{}}$\left(\begin{array}{@{}c@{}}     x^2  \\ \vdots \\      x^{2m}\\   \end{array}\right)$\end{tabular}  \\
          &                              & \multicolumn{1}{l}{}                                                                                                                                                                                                                             & \multicolumn{1}{l}{}                                                                                                                                                                                                                                                                                                                                                                                                                                                                                                                                                                                                                                                                                                         & \multicolumn{1}{l}{}                                                                                                                                                                                       \\          
T-GARCH   & $\EE[\epsilon_t^{m}]$        & \begin{tabular}[c]{@{}c@{}}$\sigma_{t} = \omega_0 + \sum\limits_{i=1}^q\big(\alpha_{0i}^+ \epsilon_{t-i}^+ + \alpha_{0i}^- \epsilon_{t-i}^-\big) ${}\\ $+ \sum\limits_{j=1}^p\beta_{0j} \sigma_{t-j}\textcolor{white}{aii}$\end{tabular}             & \begin{tabular}[c]{@{}c@{}}$\left(\begin{array}{@{}ccccc|ccc@{}}     \alpha_1^+ \eta^+ & \alpha_1^- \eta^+ & \dots & \alpha_q^+ \eta^+ & \alpha_q^- \eta^+ & \beta_1 \eta^+ & \dots & \beta_p \eta^+ \\         \alpha_1^+ \eta^- & \alpha_1^- \eta^- & \dots & \alpha_q^+ \eta^- & \alpha_q^- \eta^- & \beta_1 \eta^- & \dots & \beta_p \eta^- \\      & I_{(2q-2)\times(2q-2)} & & O_{(2q-2)\times 2} & & & O_{(2q-2)\times p} & \\\hline     \alpha_1^+  & \alpha_1^-  & \dots & \alpha_q^+  & \alpha_q^-  & \beta_1  & \dots & \beta_p  \\      & & O_{(p-1)\times 2q} &  &  & & I_{(p-1)\times(p-1)} & O_{(p-1)\times 1}\\   \end{array}\right)$\end{tabular}                                                 & \begin{tabular}[c]{@{}c@{}}$\left(\begin{array}{@{}c@{}} x^+ \\ \vdots \\ (x^+)^m \\ x^- \\ \vdots \\ (x^-)^m\\   \end{array}\right)$\end{tabular}                                    \\
          &                              & \multicolumn{1}{l}{}                                                                                                                                                                                                                             & \multicolumn{1}{l}{}                                                                                                                                                                                                                                                                                                                                                                                                                                                                                                                                                                                                                                                                                                         & \multicolumn{1}{l}{}                                                                                                                                                                                       \\
          AP-GARCH  & $\EE[\epsilon_t^{\delta m}]$ & \begin{tabular}[c]{@{}c@{}}$\sigma_{t}^\delta = \omega_0 + \sum\limits_{i=1}^q\alpha_{0i}\big(|\epsilon_{t-i}| - \gamma \epsilon_{t-i}\big)^\delta$ \\ $ + \sum\limits_{j=1}^p\beta_{0j} \sigma_{t-j}^\delta \textcolor{white}{aai}$\end{tabular}     & \begin{tabular}[c]{@{}c@{}}$\left(\begin{array}{@{}ccc|ccc@{}}     \alpha_1 (|\eta|-\gamma \eta)^\delta & \dots & \alpha_q (|\eta|-\gamma \eta)^\delta & \beta_1 (|\eta|-\gamma \eta)^\delta & \dots & \beta_p (|\eta|-\gamma \eta)^\delta \\      & I_{(q-1)\times(q-1)} & O_{(q-1)\times 1} &  & O_{(q-1)\times p} & \\\hline    \alpha_1  & \dots & \alpha_q  & \beta_1  & \dots & \beta_p  \\      & O_{(p-1)\times q} &  &  & I_{(p-1)\times(p-1)} & O_{(p-1)\times 1}\\   \end{array}\right)$\end{tabular}                                                                                                                                                                                                    & \begin{tabular}[c]{@{}c@{}}$\left(\begin{array}{@{}c@{}} (x^+)^\delta \\ \vdots \\ (x^+)^{\delta m}\\ (x^-)^\delta \\ \vdots \\ (x^-)^{\delta m}\\   \end{array}\right)$\end{tabular} \\
          &                              & \multicolumn{1}{l}{}                                                                                                                                                                                                                             & \multicolumn{1}{l}{}                                                                                                                                                                                                                                                                                                                                                                                                                                                                                                                                                                                                                                                                                                         & \multicolumn{1}{l}{}                                                                                                                                                                                       \\
GJR-GARCH & $\EE[\epsilon_t^{2m}]$       & \begin{tabular}[c]{@{}c@{}}$\sigma_{t}^2 = \omega_0 + \sum\limits_{i=1}^q\big(\alpha_{0i}^+ (\epsilon_{t-i}^+)^2 + \alpha_{0i}^- (\epsilon_{t-i}^-)^2\big)${}\\ $ + \sum\limits_{j=1}^p\beta_{0j} \sigma_{t-j}^2\textcolor{white}{aaaaaai}$\end{tabular} & \begin{tabular}[c]{@{}c@{}}$\left(\begin{array}{@{}ccccc|ccc@{}}     \alpha_1^+ (\eta^+)^2 & \alpha_1^- (\eta^+)^2 & \dots & \alpha_q^+ (\eta^+)^2 & \alpha_q^- (\eta^+)^2 & \beta_1 (\eta^+)^2 & \dots & \beta_p (\eta^+)^2 \\         \alpha_1^+ (\eta^-)^2 & \alpha_1^- (\eta^-)^2 & \dots & \alpha_q^+ (\eta^-)^2 & \alpha_q^- (\eta^-)^2 & \beta_1 (\eta^-)^2 & \dots & \beta_p (\eta^-)^2 \\      & I_{(2q-2)\times(2q-2)} & & O_{(2q-2)\times 2} & & & O_{(2q-2)\times p} & \\\hline    \alpha_1^+  & \alpha_1^-  & \dots & \alpha_q^+  & \alpha_q^-  & \beta_1  & \dots & \beta_p  \\      & & O_{(p-1)\times 2q} &  &  & & I_{(p-1)\times(p-1)} & O_{(p-1)\times 1}\\   \end{array}\right)$\end{tabular} & \begin{tabular}[c]{@{}c@{}}$\left(\begin{array}{@{}c@{}}  (x^+)^2 \\ \vdots \\ (x^+)^{2m} \\ (x^-)^2 \\ \vdots \\ (x^-)^{2m}\\   \end{array}\right)$\end{tabular}                      \\
          &                              & \multicolumn{1}{l}{}                                                                                                                                                                                                                             & \multicolumn{1}{l}{}                                                                                                                                                                                                                                                                                                                                                                                                                                                                                                                                                                                                                                                                                                         & \multicolumn{1}{l}{}                                                                                                                                                                                       \\ \hline \hline 
\end{tabular}}
\caption{Examples of GARCH-type models for which the proposed bootstrap-based test can be applied.}
\label{tab:7.1}
\end{table}

\end{landscape}
}

\begin{remark}
\label{rem:7.3}
The bootstrap-based test for the existence of moments can be readily adapted to other GARCH-type processes such as the threshold GARCH (T-GARCH) of \cite{zakoian1994threshold},  the asymmetric power GARCH (AP-GARCH) of \cite{ding1993long} or the GARCH extension of \cite{glosten1993relation} (GJR-GARCH). In fact, it only requires replacing the model-specific function $A(\theta,\eta)$ in \eqref{eq:689954}; see Table \ref{tab:7.1} for details on the functional forms of the aforementioned GARCH-type models. The theoretical results presented for the GARCH carry over after a small adjustment of the moment function $h(x)$: e.g.\ in the T-GARCH case the corresponding function is given by $h(x)=\big(x^+,\dots,(x^+)^m,x^-,\dots,(x^-)^m\big)'$, where $x^+=\max(x,0)$ and $x^-=\max(-x,0)$.\footnote{Although $h(x)$ is not differentiable at $x=0$ in the T-GARCH case, it is worth mentioning that $\nu$ in Theorem \ref{thm:7.2} is well defined as $h$ is differentiable almost everywhere.}
%Up to this point we have purely focused on the GARCH processes. It turns out that the bootstrap-based test for the existence of moments can be readily adapted to other GARCH-type processes such as the threshold GARCH (T-GARCH) of \cite{zakoian1994threshold},  the asymmetric power GARCH (AP-GARCH) of \cite{ding1993long} or the GARCH extension of \cite{glosten1993relation} (GJR-GARCH). In fact, it only requires replacing the model-specific function $A(\theta,\eta)$ in \eqref{eq:689954}. Table \ref{tab:7.2} represents the corresponding functional forms of the aforementioned GARCH-type models. The theoretical results presented carry over after a small adjustment of the moment function $h(x)$ tabulated in Table \ref{tab:7.2}. For example, in the T-GARCH case the corresponding function is given by $h(x)=\big(x^+,\dots,(x^+)^m,x^-,\dots,(x^-)^m\big)'$, where $x^+=\max(x,0)$ and $x^-=\max(-x,0)$. Although $h(x)$ is not differentiable at $x=0$ in that case, it is worth mentioning that $\nu$ in Theorem \ref{thm:7.2} is well defined as $h$ is differentiable almost everywhere.
\end{remark}

%\begin{remark}
%\label{rem:7.3}
%The bootstrap-based test for the existence of moments can be readily adapted to other GARCH-type processes. In fact, it only requires replacing the model-specific function $A(\theta,\eta)$ in \eqref{eq:689954}. For examples see Table \ref{tab:7.2} for the functional form in the threshold GARCH (T-GARCH) model. The theoretical results presented carry over after a small adjustment of the moment function $h(x)$.\footnote{Note that the functions tabulated in Table \ref{tab:7.2} are differentiable almost everywhere such that $\nu$ in Theorem \ref{thm:7.2} is well defined.}
%\end{remark}

\begin{comment}
\begin{example}
\label{ex:7.2}
Suppose $\{\epsilon_t\}$ follows a T-GARCH$(1,1)$ process given by \eqref{eq:7.2.1} and $\sigma_{t+1} = \omega_0 + \alpha_0^+ \epsilon_t^+ + \alpha_0^- \epsilon_t^- + \beta_0 \sigma_{t}$, where $\theta_0 = (\omega_0,\alpha_0^+,\alpha_0^-,\beta_0)'\in (0,\infty)\times[0,\infty)^2\times[0,1)$. The necessary and sufficient condition for the existence of $\EE[|\epsilon_t|^d]$ is 
%
\begin{align}
\label{eq:87879759975}
\sum_{i=0}^d \binom{d}{i} \Big(\big(\alpha_0^+\big)^i \mu_i^+ + \big(\alpha_0^-\big)^i \mu_i^-\Big) \beta_{0}^{m-i} \textcolor{red}{\mu_{2i}}<1
\end{align}
\end{example}
\end{comment}

\section{Numerical Illustration}
\label{sec:7.6}

\subsection{Monte Carlo Experiment}
\label{sec:7.6.1}

A simulation study is conducted to gain further insights into the practical implications of the bootstrap-based test of Section \ref{sec:7.5}. In particular we focus on the GARCH($1$,$2$) model, which is motivated by the subsequent empirical application (see Section \ref{sec:7.6.2}). The innovations are generated from a standard normal distribution, i.e.\ $\eta_t \overset{iid}{\sim}N(0,1)$, such that $(\mu_4,\mu_6,\mu_8,\mu_{10})=(3,15,105,945)$. Further, the GARCH parameters are set to $\omega_0 = 0.08$, $\alpha_{01} = 0.05$ and $\alpha_{02} = 0.10$ while $\beta_{01} \approx 0.80$ is chosen such that $T$ in \eqref{eq:8787521} is equal to unity when $m=3$. In other words, $m=3$ corresponds to the boundary case of the null hypothesis in which  $\EE[\epsilon_t^{2m}]$ is just evaluated infinite.
%In other words, $\EE[\epsilon_t^{2m}]$ exists for $m=1,2$ for the GARCH($1$,$2$) process, but does not for $m=4,5,\dots$. The case $m=3$ corresponds to the boundary case of the null hypothesis in which  $\EE[\epsilon_t^{2m}]$ is just evaluated infinite.
We consider three estimation sample sizes, $n \in \{1{,}000; 5{,}000; 10{,}000\}$, whereas the number of bootstrap replicates is fixed and equal to $B=1{,}999$. For each model version we simulate $S=2{,}000$ independent Monte Carlo trajectories and investigate the proposed bootstrap test at two nominal levels: $5\%$ and $10\%$.

\begin{figure}[ht!]
%\captionsetup[subfigure]{labelformat=roman}
\centering
\begin{subfigure}[b]{0.45\textwidth}
\label{fig:1a}
	\centering
	\includegraphics[width=\textwidth]{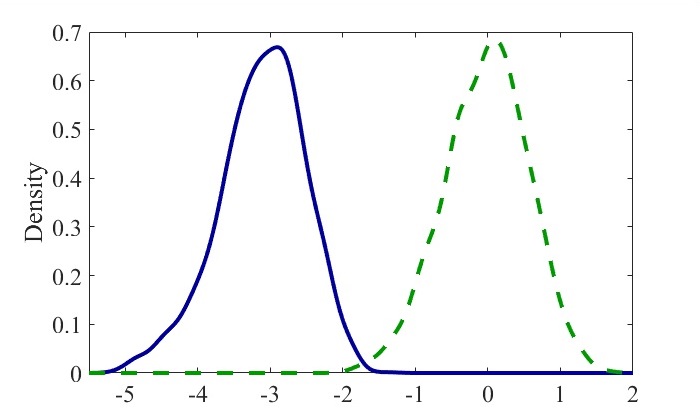}
            \caption{$m=1$ with $T=0.96$}    
\end{subfigure}
\quad
\begin{subfigure}[b]{0.45\textwidth}  
\label{fig:1b}
	\centering 
	\includegraphics[width=\textwidth]{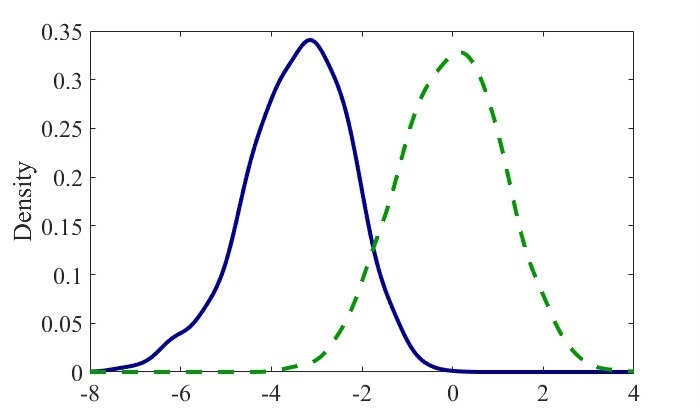}
	\caption{$m=2$ with $T=0.95$} 
\end{subfigure}
\begin{subfigure}[b]{0.45\textwidth}
\label{fig:1c}
	\centering
	\includegraphics[width=\textwidth]{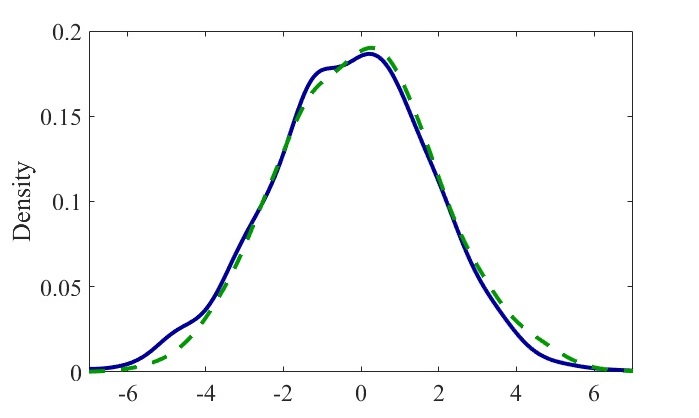}
            \caption{$m=3$ with $T=1.00$}    
\end{subfigure}
\quad
\begin{subfigure}[b]{0.45\textwidth}  
\label{fig:1d}
	\centering 
	\includegraphics[width=\textwidth]{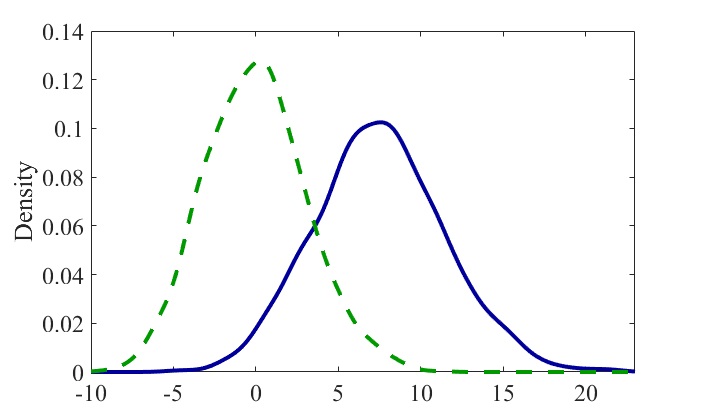}
	\caption{$m=4$ with $T=1.11$} 
\end{subfigure}
\begin{subfigure}[b]{0.45\textwidth}
\label{fig:1e}
	\centering
	\includegraphics[width=\textwidth]{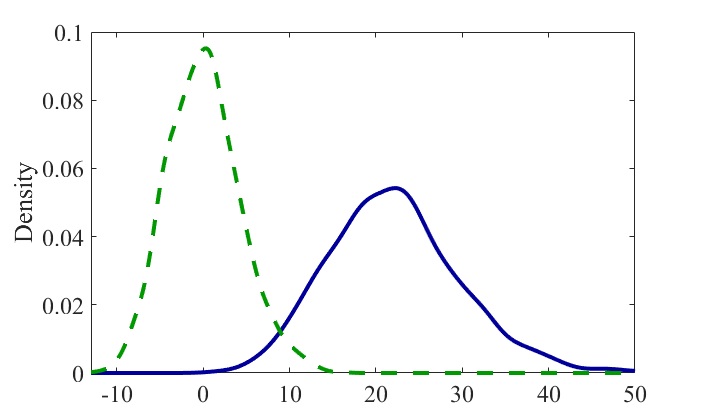}
            \caption{$m=5$ with $T=1.32$}    
\end{subfigure}
\caption{Density estimates for the distribution of $\sqrt{n}(\hat{T}_n-1)$ (solid blue line) based on $S=2{,}000$ simulations and the bootstrap distribution of $\sqrt{n}(\hat{T}_n^\star-\hat{T}_n^c)$ (dashed green line) based on $B=1{,}999$. The data generating process is a GARCH($1$,$2$) with Gaussian innovations and sample size $n=5{,}000$.} 
\label{blubb}
\end{figure}

Figure \ref{blubb} displays the density of the distribution of $\sqrt{n}(\hat{T}_n-1)$ and the bootstrap distribution of $\sqrt{n}(\hat{T}_n^\star-\hat{T}_n^c)$ for varying $m$ and sample size $n=5{,}000$. For $m=1$ and $m=2$ one observes that the two densities have a similar shape. The key difference is that the bootstrap distribution is centered around zero, whereas the distribution of $\sqrt{n}(\hat{T}_n-1)$ is shifted to the left (as expected) with center $\sqrt{n}(T-1)$, i.e.\ $-3.00$ for $m=1$ and $-3.22$ for $m=2$. For the case $m=3$, which corresponds to the boundary of the null hypothesis, Figure \ref{blubb}(iii) shows that the bootstrap distribution of $\sqrt{n}(\hat{T}_n^\star-\hat{T}_n^c)$ mimics well the finite sample distribution of $\sqrt{n}(\hat{T}_n-1)$. For $m=4$ and $m=5$, the null hypothesis is violated and the bootstrap and the non-bootstrap distribution exhibit distinct behavior as visualized in Figures \ref{blubb}(d) and \ref{blubb}(e). Whereas the bootstrap distribution remains centered around the origin, the  distribution of $\sqrt{n}(\hat{T}_n-1)$ is more disperse and starts to diverge with center $\sqrt{n}(T-1)$, i.e.\ $7.83$ for $m=4$ and $22.40$ for $m=5$.

%The simulations are carried out in Matlab R2016a on $16$ cores of a HR Z640 workstation employing parallel computing.
Table \ref{tab:7.2} reports the simulated rejection rates (in $\%$). For $m=1,2$ the null hypothesis of $\EE[\epsilon_t^{2m}]<\infty$ is (almost) never rejected by the bootstrap test at the considered nominal values across  sample sizes. For $m=3$, the relative rejection frequencies are below the corresponding nominal values, yet approach them with increasing sample size. This result suggests that the bootstrap test is rather conservative. For $m=4$ the the relative rejection frequency considerably increase (especially in larger samples) indicating that the null hypothesis is violated. For $m=5$ the results are more pronounced and the relative rejection rates are considerably higher reaching $100\%$ when the sample size is $n=10{,}000$. 

\begin{table}[bh]
\centering
\begin{tabular}{rrccccc}
\hline \hline
\multicolumn{1}{c}{Sample} & \multicolumn{1}{c}{Nominal} & $m=1$    & $m=2$    & $m=3$  & $m=4$    & $m=5$    \\
\multicolumn{1}{c}{size}   & \multicolumn{1}{c}{level}   & $T=0.96$ & $T=0.95$ & $T=1.00$  & $T=1.11$ & $T=1.32$ \\ \hline
\multicolumn{1}{l}{}       & \multicolumn{1}{c}{}        &          &          &        &          &          \\
$1{,}000$                  
                           & $5\%$                       &  $0.00$  &  $0.00$  & $2.70$ &  $18.75$ & $43.20$  \\
                           & $10\%$                      &  $0.00$  &  $0.05$  & $6.40$ & $29.35$  & $58.05$  \\
$5{,}000$                  
                           & $5\%$                       &  $0.00$  &  $0.00$  & $3.10$ & $66.80$  &  $97.95$ \\
                           & $10\%$                      &  $0.00$  &  $0.00$  & $6.50$ &  $79.25$ & $99.00$  \\
$10{,}000$                 
                           & $5\%$                       &  $0.00$  &  $0.00$  & $4.15$ & $91.15$  & $99.95$  \\
                           & $10\%$                      &  $0.00$  &  $0.00$  & $8.70$ & $95.75$  & $100.00$ \\ \hline \hline
\end{tabular}
\caption{The table reports the relative rejection frequency (in $\%$) of the null hypothesis $\EE[\epsilon_t^{2m}]<\infty$ for different sample sizes ($n$) and for different nominal levels. The bootstrap test is based on $B=1{,}999$ bootstrap replications and the rejection frequencies are computed using $S=2{,}000$ simulations. The data generating process is a GARCH($1$,$2$) with Gaussian innovations.}
\label{tab:7.2}
\end{table}

\subsection{Empirical Application}
\label{sec:7.6.2}

Next, we study the German stock market index DAX for the period January 2, 1990 until January 20, 2009. The information on the index price is retrieved from Yahoo Finance and daily (log-) returns (expressed in $\%$) are determined yielding $n=4{,}807$ observations.
%The index prices are retrieved from Yahoo Finance and daily (log-) returns (expressed in $\%$) are determined using $\epsilon_t = 100\log(p_t/p_{t-1})$, where $p_t$ denotes the closing price of the index at trading day $t$. 
%
\begin{figure}[h]
\centering
\begin{subfigure}[b]{0.45\textwidth}
\label{fig:2a}
	\centering
	\includegraphics[width=\textwidth]{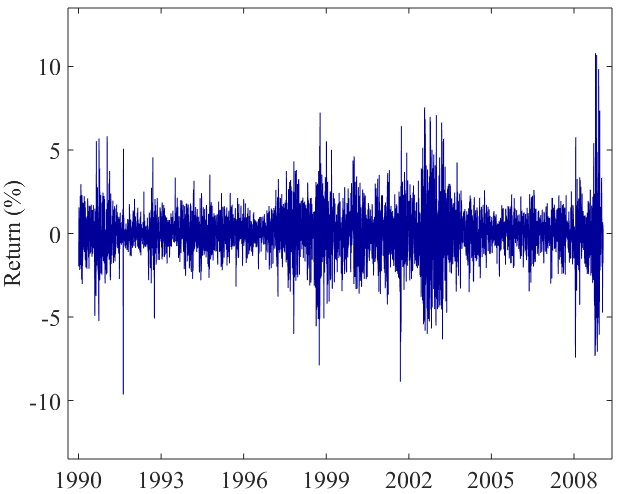}
            \caption{Returns of DAX}    
\end{subfigure}
\quad
\begin{subfigure}[b]{0.45\textwidth}  
\label{fig:2b}
	\centering 
	\includegraphics[width=\textwidth]{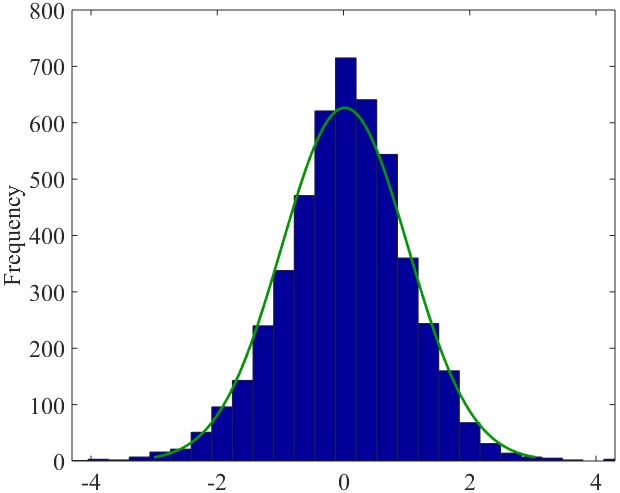}
	\caption{Histogram of the residuals $\hat{\eta}_t$'s} 
\end{subfigure}
\caption{The returns of the German stock market index DAX are plotted in (i) for the period January 2, 1990 -- January 20, 2009. The histogram of the residuals is plotted in (ii) after fitting a GARCH($1,2$) model. A scaled normal density is superimposed.} 
\label{fig:7.1}
\end{figure}
Figure \ref{fig:7.1}(i) displays the resulting series of returns. For this financial series \citeauthor{francq2011garch} (\citeyear{francq2011garch}, p.\ 206) strongly reject the null hypothesis of a GARCH($1$,$1$) in favor for a GARCH($1$,$2$) model. Estimating the latter, we present the corresponding point estimates in Table \ref{tab:7.3}, where the reported standard errors are obtained by means of bootstrap.
\begin{table}[h]
\centering
\begin{tabular}{lccccccc}
\hline \hline
                     & $\hat{\omega}_n$ & $\hat{\alpha}_{n,1}$ & $\hat{\alpha}_{n,2}$ & $\hat{\beta}_{1,n}$ & $\hat{\mu}_{n,4}$ & $\hat{\mu}_{n,6}$  \\ \hline
point estimate       & $0.0489$         & $0.0181$           & $0.0979$           & $0.8589$  & $7.9938$ & $629.9851$ \\
std. error & $0.0131$ &  $0.0213$  & $0.0293$  &  $0.0244$  &  $3.0543$ &  $421.4228$                 \\ \hline \hline
\end{tabular}
\caption{GARCH($1,2$) estimates and the estimates of the innovations' fourth and sixth moments. The standard errors are obtained by applying the fixed-design residual bootstrap with $B=9{,}999$ bootstrap replications.}
\label{tab:7.3}
\end{table}
%
% \begin{table}[h]
% \centering
% \begin{tabular}{lccccccc}
% \hline \hline
%                      & $\hat{\omega}_n$ & $\hat{\alpha}_{n,1}$ & $\hat{\alpha}_{n,2}$ & $\hat{\beta}_{1,n}$ & $\hat{\mu}_{n,4}$ & $\hat{\mu}_{n,6}$  \\ \hline
% point estimate       & $0.0491$         & $0.0182$           & $0.0983$           & $0.8584$  & $7.9927$ & $629.5224$ \\
% std. error & $0.0131$ &  $0.0214$  & $0.0298$  &  $0.0247$  &  $2.9970$ &  $415.6760$                 \\ \hline \hline
% \end{tabular}
% \caption{GARCH($1,2$) estimates and the estimates of the innovations' fourth and sixth moments. The standard errors are obtained by applying the fixed-design residual bootstrap with $B=1{,}999$ bootstrap replications.}
% \label{tab:7.3}
% \end{table}
%
Indeed we find a substantial point estimate for $\alpha_{0,2}$. Moreover, as documented in various studies we observe large volatility persistence in the data. The estimates of the fourth and sixth moments indicate that the innovation distribution is considerably more heavy-tailed than the standard normal distribution whose corresponding moments are $3$ and $15$, respectively.
Although this can be hardly seen from the histogram of the residuals in Figure \ref{fig:7.1}(ii), where a scaled normal distribution is superimposed,
%Figure \ref{fig:4.2} (b) plots the histogram of the residuals with the normal distribution superimposed.
%However,
we find that a (normalized) Student-t distribution with $9$ degrees of freedom provides an improved fit. Next, we test to what extend the financial time series at hand has finite moments. In particular we focus on the second, fourth and sixth moment corresponding to $m=1,2,3$, respectively. Table \ref{tab:7.4} presents the test-statistic and the corresponding p-value associated with the null hypothesis $\EE[\epsilon_t^{2m}]<\infty$. 
\begin{table}[h]
\centering
\begin{tabular}{lccccccc}
\hline \hline
                     & $m=1$ & $m=2$ & $m=3$ & \\ \hline
$\hat{T}_n$       & $0.9773$         & $1.0309$           & $1.5788$    \\
$\hat{p}_{n,B}^\star$ & $0.9927$ &  $0.1785$  & $0.0239$  \\ \hline \hline
\end{tabular}
\caption{Test-statistic and the corresponding p-value associated with the null hypothesis $\EE[\epsilon_t^{2m}]<\infty$. The p-value is based on $B=9{,}999$ bootstrap replications.}
\label{tab:7.4}
\end{table}
%
% \begin{table}[h]
% \centering
% \begin{tabular}{lccccccc}
% \hline \hline
%                      & $m=1$ & $m=2$ & $m=3$ & \\ \hline
% $\hat{T}_n$       & $0.9773$         & $1.0309$           & $1.5788$    \\
% $\hat{p}_{n,B}^\star$ & $0.9968$ &  $0.1726$  & $0.0005$  \\ \hline \hline
% \end{tabular}
% \caption{Test-statistic and the corresponding p-value associated with the null hypothesis $\EE[\epsilon_t^{2m}]<\infty$. The p-value is based on $B=1{,}999$ bootstrap replications.}
% \label{tab:7.4}
% \end{table}
%
For $m=1$, we find a test statistic smaller than unity and henceforth the corresponding p-value is large. For $m=2$, the test statistic is slightly larger than unity, however there is not enough evidence to reject the null hypothesis that the fourth moment exists. In contrast, for $m=3$, the test statistic is substantial larger and the corresponding p-value indicates that it is unlikely that the sixth moment is finite. Summing up: while the series seems to admit moments of second-order, there is strong evidence against the existence of sixth-order moments. With regard to the fourth-order moment, the test is inconclusive.

%The corresponding test statistic is given by $\hat{T}_{n} =  h(\hat{\theta}_n,\hat{\mu}_n)=\dots$. Plot bootstrap density 

%Report a figure like Figure 3 of \cite{francq2018testing}

\section{Concluding Remarks}
\label{sec:7.7}

This paper studies the joint inference on conditional volatility parameters and the innovation moments by means of bootstrap to test for the existence of moments for GARCH processes. For a general class of volatility models we derive the joint asymptotic distribution of the QML estimators and the empirical moments of the residuals. Further, we propose a fixed-design residual bootstrap to mimic  the estimators' finite sample distribution. The validity of the bootstrap method is proven under mild assumptions and a bootstrap-based test for the existence of moments in the GARCH($p$,$q$) model is proposed. This testing problem is non-standard as the test-statistic involves the spectral radius. Still the testing procedure is simple to implement and provides asymptotically correctly-sized tests without losing its consistency property. A simulation study demonstrates the test's size and power properties in finite samples. An empirical application illustrates the bootstrap-based testing approach, which can easily be extended to other GARCH-type settings.

\appendix

\section{Auxiliary Results and Proofs}

%\begin{proof}
\textit{Proof of Theorem \ref{thm:7.2}.}
We define $\hat{\mu}_{n,k}= \frac{1}{n}\sum_{t=1}^n \hat{\eta}_t^k$ for $k\in\{2,\dots,2m\}$ and expand
\begin{align*}
\sqrt{n}\big(\hat{\mu}_{n,k}-\mu_k) =\underbrace{\frac{1}{\sqrt{n}}\sum_{t=1}^n \epsilon_t^k \bigg(\frac{1}{\tilde{\sigma}_t^k(\hat{\theta}_n)}-\frac{1}{\sigma_t^k(\hat{\theta}_n)}\bigg)}_{I}+ \underbrace{\frac{1}{\sqrt{n}}\sum_{t=1}^n  \bigg(\frac{\epsilon_t^k}{\sigma_t^k(\hat{\theta}_n)}-\mu_k\bigg)}_{II}.
\end{align*}
Using the inequality $\big|(x+y)^k-x^k\big|\leq k2^{k-1}|y|\big(|x|^{k-1}+|y|^{k-1}\big)$ for $x, y \in \R$ and Assumptions \ref{as:7.3} and \ref{as:7.4}(\ref{as:7.4.1}) leads to
\begin{align*}
&\sup_{\theta \in \Theta}\bigg|\frac{1}{\tilde{\sigma}_t^k(\theta)}-\frac{1}{\sigma_t^k(\theta)}\bigg|=\sup_{\theta \in \Theta}\frac{|\tilde{\sigma}_t^k(\theta)-\sigma_t^k(\theta)|}{\tilde{\sigma}_t^k(\theta)\sigma_t^k(\theta)}\\
\leq& k2^{k-1} \sup_{\theta \in \Theta}\frac{\big|\tilde{\sigma}_t(\theta)-\sigma_t(\theta)\big| \Big(\sigma_t^{k-1}(\theta)+\big|\tilde{\sigma}_t(\theta)-\sigma_t(\theta)\big|^{k-1}\Big)}{\tilde{\sigma}_t^k(\theta)\sigma_t^k(\theta)}\\
\leq& k2^{k-1}\bigg(\frac{C_1\rho^t}{\underline{\omega}^{k+1}}+\frac{C_1^k\rho^{tk}}{\underline{\omega}^{2k}}\bigg)\leq k2^{k-1}\bigg(\frac{C_1}{\underline{\omega}^{k+1}}+\frac{C_1^k}{\underline{\omega}^{2k}}\bigg)\rho^t
\end{align*}
such that
\begin{align*}
|I|\leq k2^{k-1}\bigg(\frac{C_1}{\underline{\omega}^{k+1}}+\frac{C_1^k}{\underline{\omega}^{2k}}\bigg)\frac{1}{\sqrt{n}}\sum_{t=1}^n \rho^t|\epsilon_t|^k.
\end{align*}
For each $\varepsilon>0$, Markov's inequality and the $c_r$-inequality entail 
\begin{align*}
& \PP\bigg[\frac{1}{\sqrt{n}}\sum_{t=1}^n\rho^{t}|\epsilon_t|^k>\varepsilon\bigg] \leq \frac{1}{(\sqrt{n}\varepsilon)^{s/k}} \EE\Bigg[\bigg(\sum_{t=1}^n \rho^{t} |\epsilon_t|^k\bigg)^{s/k}\Bigg]\\
%%%
\leq& \frac{1}{(\sqrt{n}\varepsilon)^{s/k}}\sum_{t=1}^n \rho^{st/k} \EE\big[|\epsilon_t|^s\big] \leq \frac{1}{(\sqrt{n}\varepsilon)^{s/k}} \frac{\EE[|\epsilon_t|^s]}{\varepsilon^{s/k}(1-\rho^{s/k})}\to 0
\end{align*}
since $\rho \in(0,1)$ and $\EE[|\epsilon_t|^s]=\EE[\sigma_t^s]\EE[|\eta_t|^s]<\infty$ for some $s\in(0,1]$ by Assumptions \ref{as:7.3} and \ref{as:7.5}(\ref{as:7.5.1}). Hence, we have $|I| \overset{p}{\to}0$. Regarding $II$, a Taylor expansion yields
%
% \begin{align*}
% \frac{\tilde{\sigma}_t^k(\hat{\theta}_n)}{\tilde{\sigma}_t^k(\hat{\theta}_n^*)}=
% %%%
%  1-kD_t'\big(\hat{\theta}_n-\theta_0\big)+\frac{1}{2}\big(\hat{\theta}_n-\theta_0\big)'\frac{\sigma_t^k(\theta_0)}{\sigma_t^k(\bar{\theta}_n)}\Big(k(k+1)D_t(\bar{\theta}_n)D_t'(\bar{\theta}_n)-kH_t(\bar{\theta}_n)\Big) \big(\hat{\theta}_n-\theta_0\big),
% \end{align*}
%
%
\begin{align*}
II =& \frac{1}{\sqrt{n}}\sum_{t=1}^n  \big(\eta_t^k-\mu_k\big)-\frac{1}{n}\sum_{t=1}^n  k\eta_t^k D_t'\sqrt{n}\big(\hat{\theta}_n-\theta_0\big)\\
&+\frac{1}{2}\sqrt{n}\big(\hat{\theta}_n-\theta_0\big)'\frac{1}{n}\sum_{t=1}^n  \frac{\sigma_t^k(\theta_0)}{\sigma_t^k(\bar{\theta}_n)}\Big(k(k+1)D_t(\bar{\theta}_n)D_t'(\bar{\theta}_n)-kH_t(\bar{\theta}_n)\Big)\eta_t^k \big(\hat{\theta}_n-\theta_0\big),
\end{align*}
where $H_t(\theta)= \frac{1}{\sigma_t(\theta)}\frac{\partial^2 \sigma_t(\theta)}{\partial \theta \partial \theta'}$ and  $\bar{\theta}_n$ lies between $\hat{\theta}_n$ and $\theta_0$. The last term vanishes in probability since $\sqrt{n}\big(\hat{\theta}_n-\theta_0\big)=O_p(1)$ \citeauthor{francq2015risk} (\citeyear{francq2015risk}, Theorem 2) and 
\begin{align*}
&\bigg|\bigg|\frac{1}{n}\sum_{t=1}^n  \frac{\sigma_t^k(\theta_0)}{\sigma_t^k(\bar{\theta}_n)}\Big(k(k+1)D_t(\bar{\theta}_n)D_t'(\bar{\theta}_n)-kH_t(\bar{\theta}_n)\Big)\eta_t^k\bigg|\bigg|\\
%%%
\overset{a.s.}{\leq} &\frac{1}{n}\sum_{t=1}^n \sup_{\theta \in \mathscr{V}(\theta_0)} \frac{\sigma_t^k(\theta_0)}{\sigma_t^k(\theta)}\Big(k(k+1)||D_t(\theta)||^2+k||H_t(\theta)||\Big)|\eta_t|^k\\
%%%
\overset{a.s.}{\to}&\EE\bigg[ \sup_{\theta \in \mathscr{V}(\theta_0)} \frac{\sigma_t^k(\theta_0)}{\sigma_t^k(\theta)}\Big(k(k+1)||D_t(\theta)||^2+k||H_t(\theta)||\Big)|\eta_t|^k\bigg]\\
%%%
\leq& \EE\big[S_t^{2k}\big]\Big(k(k+1)\EE\big[U_t^{4}\big]+k\EE\big[V_t^{2}\big]\Big)\EE\big[|\eta_t|^k\big]<\infty
\end{align*}
with $S_t = \sup_{\theta \in \mathscr{V}(\theta_0)} \frac{\sigma_t(\theta_0)}{\sigma_t(\theta)}$, $U_t = \sup_{\theta \in \mathscr{V}(\theta_0)}||D_t(\theta)||$ and $V_t = \sup_{\theta \in \mathscr{V}(\theta_0)}||H_t(\theta)||$, where we used the uniform ergodic theorem  (c.f.\ \citeauthor{francq2011garch}, \citeyear{francq2011garch}, Theorem A.2 and p.\ 181). Further, the ergodic theorem implies $\frac{1}{n}\sum_{t=1}^n  k\eta_t^kD_t\overset{a.s.}{\to}k\mu_k \Omega$. Combining results, we have
\begin{align*}
\sqrt{n}\big(\hat{\mu}_{n}-\mu\big) = \frac{1}{\sqrt{n}}\sum_{t=1}^n\big(h(\eta_t)-\mu\big)-\nu \Omega'\sqrt{n}\big(\hat{\theta}_n-\theta_0\big)+o_p(1).
\end{align*}
Inserting the expansion for $\sqrt{n}\big(\hat{\theta}_n-\theta_0\big)$ given in \cite{francq2015risk}, i.e.\
\begin{align*}
\sqrt{n}\big(\hat{\theta}_n-\theta_0\big) = \frac{1}{2}J^{-1} \frac{1}{\sqrt{n}}\sum_{t=1}^n D_t\big(\eta_t^{2}-1\big)+o_{p}(1),
\end{align*}
we establish
\begin{align*}
      \begin{pmatrix}
      \sqrt{n}(\hat{\theta}_n-\theta_0)\\
    \sqrt{n}(\hat{\mu}_{n} - \mu)
      \end{pmatrix}
      = 
      \begin{pmatrix}
      \frac{1}{2}J^{-1}&O_{m\times m}\\
    -\frac{1}{2}\nu\Omega'J^{-1} & I_{m\times m}
      \end{pmatrix}
        \begin{pmatrix}
      \frac{1}{\sqrt{n}}\sum_{t=1}^n D_t\big(\eta_t^{2}-1\big)\\
\frac{1}{\sqrt{n}}\sum_{t=1}^n\big(h(\eta_t)-\mu\big)
      \end{pmatrix}+o_{p}(1).
\end{align*}
The Wold-Cràmer device and the central limit theorem for martingale differences (c.f.\ \citeauthor{francq2011garch}, \citeyear{francq2011garch}, Corollary A.1) implies
\begin{align*}
  \frac{1}{\sqrt{n}}\sum_{t=1}^n    \begin{pmatrix}
       D_t \big(\eta_t^{2}-1\big)\\
   h(\eta_t)-\mu
      \end{pmatrix}
\overset{d}{\to}N (0,\Psi)
      \quad \text{with}\quad \Psi =\begin{pmatrix}
     (\kappa-1)J&  \Omega \xi'\\
    \xi \Omega' & \Upsilon.
     \end{pmatrix}
\end{align*}
The result follows noting that
\begin{align*}
      \begin{pmatrix}
      \frac{1}{2}J^{-1}&O_{m\times m}\\
    -\frac{1}{2}\nu\Omega'J^{-1} & I_{m\times m}
      \end{pmatrix}
\Psi
     \begin{pmatrix}
      \frac{1}{2}J^{-1}&O_{m\times m}\\
    -\frac{1}{2}\nu\Omega'J^{-1} & I_{m\times m}
      \end{pmatrix}' = 
      \begin{pmatrix}
      \frac{\mu_4-1}{4}J^{-1} & -J^{-1}\Omega \nu'\\
    -\nu \Omega'J^{-1} & \Xi
      \end{pmatrix},
\end{align*}
which completes the proof. \qed
%\end{proof}

\begin{lemma}
\label{lem:7.2}
Under Assumptions \ref{as:7.1}--\ref{as:7.4}, \ref{as:7.5}(\ref{as:7.5.1}), \ref{as:7.5}(\ref{as:7.5.3}), \ref{as:7.6}, \ref{as:7.9} and \ref{as:7.10} with $a=-1,\max\{4,2d\}$, $b=4$ and $c=2$, we have
\begin{align*}
  \frac{1}{\sqrt{n}}\sum_{t=1}^n    \begin{pmatrix}
       \hat{D}_t \big(\eta_t^{*2}-1\big)\\
   h(\eta_t^*)-\hat{\mu}_n
      \end{pmatrix}
\overset{d^*}{\to}N (0,\Psi)
      \quad \text{with}\quad \Psi =\begin{pmatrix}
     (\kappa-1)J&  \Omega \xi'\\
    \xi \Omega' & \Upsilon
     \end{pmatrix}
\end{align*}
almost surely.
\end{lemma}
\begin{proof}
\citeauthor{beutner2018residual} (\citeyear{beutner2018residual}, proof of Lemma 7) shows that $\frac{1}{\sqrt{n}} \sum_{t=1}^n \hat{D}_t \big(\EE^*[\eta_t^{*2}]-1\big)=0$ for sufficiently large $n$ almost surely since $\hat{\theta}_n\overset{a.s.}{\to}\theta_0 \in \mathring{\Theta}$ and $\EE^*\big[\eta_t^{*2}\big]=1$ whenever $\hat{\theta}_n \in \mathring{\Theta}$ under Assumption \ref{as:7.10}. It remains to show that for each  $\lambda=(\lambda_1',\lambda_2)'  \in \R^{r+dim(\mu)}$ with $||\lambda||\neq 0$
\begin{align*}
 \sum_{t=1}^n \underbrace{\frac{1}{\sqrt{n}}\lambda'  \begin{pmatrix}
       \hat{D}_t \big(\eta_t^{*2}-\EE^*[\eta_t^{*2}]\big)\\
   h(\eta_t^*)-\hat{\mu}_n
      \end{pmatrix}}_{Z_{n,t}^*}
    \overset{d^*}{\to}N \left(0,\lambda'\Psi \lambda\right)
\end{align*}
almost surely by the Cram\'er-Wold device. By construction, we have $\EE^*\big[Z_{n,t}^*\big]=0$. Further, we have that $s_n^2 =\sum_{t=1}^n\Var^*\big[Z_{n,t}^{*}\big]$ is equal to
\begin{align}
\lambda'\begin{pmatrix}
      \Var^*[\eta_t^{*2}]\hat{J}_n&  \hat{\Omega}_n\Cov^*[\eta_t^{*2},h(\eta_t^*)]'\\
    \Cov^*[\eta_t^{*2},h(\eta_t^*)] \hat{\Omega}_n' & \Var^*[h(\eta_t^*)] 
      \end{pmatrix}\lambda.
\end{align}
\citeauthor{beutner2018residual} (\citeyear{beutner2018residual}, Lemma 2) gives $\hat{J}_n \overset{a.s.}{\to}J$ and $\hat{\Omega}_n\overset{a.s.}{\to}\Omega$. Further, \citeauthor{beutner2018residual} (\citeyear{beutner2018residual}, Lemma 5) yields $\EE^*[\eta_t^{*k}]\overset{a.s.}{\to}\EE[\eta_t^{k}]$ for $k=1,\dots,d$ implying $\Cov^*[\eta_t^{*2},h(\eta_t^*)]\overset{a.s.}{\to}\Cov[\eta_t^{2},h(\eta_t)]=\xi$ and $\Var^*[h(\eta_t^*)]\overset{a.s.}{\to}\Var[h(\eta_t)]=\Upsilon$. Thus, we get $s_n^2\overset{a.s.}{\to} \lambda'\Psi \lambda$. Next, we verify Lindeberg condition. For any $\varepsilon>0$ 
\begin{align*}
&  \sum_{t=1}^n  \EE^*\big[Z_{n,t}^{*2}\mathbbm{1}_{\{|Z_{n,t}^{*}|\geq s_n \varepsilon \}}\big] \leq  \underbrace{\sum_{t=1}^n\EE^*\big[Z_{n,t}^{*2}\mathbbm{1}_{\{|\eta_t^{*}|> C \}}\big]}_{I}+ \underbrace{\sum_{t=1}^n  \EE^*\big[Z_{n,t}^{*2}\mathbbm{1}_{\{|Z_{n,t}^{*}|\geq s_n \varepsilon \}}\mathbbm{1}_{\{|\eta_t^{*}|\leq C \}}\big]}_{II}
\end{align*}
holds, where $C>0$. Employing the elementary inequality $\Big(\sum_{i=1}^N x_i\Big)^2\leq N \sum_{i=1}^N x_i^2$ for all $x_1,\dots,x_N \in \R$ and $N \in \N$ we find that
\begin{align*}
Z_{n,t}^{*2} \leq & \frac{4}{n} \Big(\big(\lambda_1'\hat{D}_t\big)^2 \big(\eta_t^{*4}+\EE^*[\eta_t^{*2}]^2\big)+
    ||\lambda_2||^2 \big(||h(\eta_t^{*})||^2+||\hat{\mu}_{n}||^2\big)\Big).
\end{align*}
Thus, we obtain
\begin{align*}
I \leq &  \frac{4}{n}\sum_{t=1}^n\EE^*\bigg[ \Big(\big(\lambda_1'\hat{D}_t\big)^2 \big(\eta_t^{*4}+\EE^*[\eta_t^{*2}]^2\big)+
    ||\lambda_2||^2 \big(||h(\eta_t^{*})||^2+||\hat{\mu}_{n}||^2\big)\Big)\mathbbm{1}_{\{|\eta_t^{*}|> C \}}\bigg]\\
%%%
= & 8 \Big( \lambda_1' \hat{J}_n \lambda_1 \EE^*\big[\eta_t^{*4}\mathbbm{1}_{\{|\eta_t^{*}|> C \}}\big]+||\lambda_2||^2 \EE^*\big[||h(\eta_t^{*})||^2\mathbbm{1}_{\{|\eta_t^{*}|> C \}}\big]\\
&\qquad +\big(\lambda_1' \hat{J}_n \lambda_1 \EE^*[\eta_t^{*2}]^2+
    \lambda_2^2||\hat{\mu}_{n}||^2\big)\EE^*\big[\mathbbm{1}_{\{|\eta_t^{*}|> C \}}\big]\Big)\\
    \overset{a.s.}{\to}&8 \Big( \lambda_1' J \lambda_1 \EE\big[\eta_t^{4}\mathbbm{1}_{\{|\eta_t|> C \}}\big]+||\lambda_2||^2 \EE\big[||h(\eta_t)||^2\mathbbm{1}_{\{|\eta_t|> C \}}\big]\\
&\qquad +\big(\lambda_1' J \lambda_1 \EE[\eta_t^{2}]^2+
    \lambda_2^2||\hat{\mu}_{n}||^2\big)\EE\big[\mathbbm{1}_{\{|\eta_t|> C \}}\big]\Big)
\end{align*}
and choosing $C$ sufficiently large yields $I\overset{a.s.}{\to}0$. Given a value of $C$, we have
\begin{align*}
II\leq  &  \frac{4}{n}  \sum_{t=1}^n \EE^*\bigg[ \Big(\big(\lambda_1'\hat{D}_t\big)^2 \big(\eta_t^{*4}+\EE^*[\eta_t^{*2}]^2\big)+
    ||\lambda_2||^2 \big(||h(\eta_t^{*})||^2+||\hat{\mu}_{n}||^2\big)\Big)\\
    & \qquad \qquad \quad \times \mathbbm{1}_{ \{ ||\lambda_1|| (\eta_t^{*2}+\EE^*[\eta_t^{*2}]) \max_{t} ||\hat{D}_t||+
    ||\lambda_2||(||h(\eta_t^*)||+||\hat{\mu}_{n}||)\geq \sqrt{n} s_n \varepsilon \}}\mathbbm{1}_{\{|\eta_t^{*}|\leq C \}}\bigg]\\
%%%
\leq  &  \frac{4}{n}  \sum_{t=1}^n  \Big(\big(\lambda_1'\hat{D}_t\big)^2 \big(C^4+\EE^*[\eta_t^{*2}]^2\big)+
    \lambda_2^2 \big(||h(C)||^2+||\hat{\mu}_{n}||^2\big)\Big)\\
    & \qquad \qquad \quad \times \mathbbm{1}_{ \{ ||\lambda_1|| (C^2+\EE^*[\eta_t^{*2}]) \max_{t} ||\hat{D}_t||+
    ||\lambda_2||(||h(C)||+||\hat{\mu}_{n}||)\geq \sqrt{n} s_n \varepsilon \}}\\
%%%
\leq  &  4  \Big(\lambda_1' \hat{J}_n \lambda_1 \big(C^4+\EE^*[\eta_t^{*2}]^2\big)+
    \lambda_2^2 \big(||h(C)||^2+||\hat{\mu}_{n}||^2\big)\Big)\\
    & \qquad \qquad \quad \times \mathbbm{1}_{ \{ ||\lambda_1|| (C^2+\EE^*[\eta_t^{*2}]) \max_{t} ||\hat{D}_t||+
    |\lambda_2|(||h(C)||+||\hat{\mu}_{n}||)\geq \sqrt{n} s_n \varepsilon \}}\\    
    %%%
    \overset{a.s.}{\to}& 4  \Big(\lambda_1' J \lambda_1 \big(C^4+\EE[\eta_t^{2}]^2\big)+
    \lambda_2^2 \big(||h(C)||^2+||\hat{\mu}_{n}||^2\big)\Big)\times 0 = 0
\end{align*}
as $\max_t||\hat{D}_t||/\sqrt{n}\overset{a.s.}{\to}0$. Combining results, gives $\frac{1}{s_n^2}\sum_{t=1}^n  \EE^*\big[Z_{n,t}^{*2}\mathbbm{1}_{\{|Z_{n,t}^{*}|\geq s_n \epsilon \}}\big] \overset{a.s.}{\to}0$.
 The Central Limit Theorem for triangular arrays (c.f.\ \citeauthor{billingsley1986probability}, \citeyear{billingsley1986probability}, Theorem 27.3) implies that $\sum_{t=1}^n Z_{n,t}^*$ converges in conditional distribution to $N\big(0,\lambda'\Psi\lambda \big)$ almost surely, which completes the proof. 
\end{proof}

%\begin{proof}
\noindent
\textit{Proof of Theorem \ref{thm:7.3}.}
We define $\hat{\mu}_{n,k}^*= \frac{1}{n}\sum_{t=1}^n \hat{\eta}_t^{*k}$ for $k\in\{2,\dots,2m\}$; a Taylor expansion yields
\begin{align*}
&\sqrt{n}\big(\hat{\mu}_{n,k}^*-\hat{\mu}_{n,k}) = \frac{1}{\sqrt{n}}\sum_{t=1}^n  \bigg(\frac{\tilde{\sigma}_t^k(\hat{\theta}_n)}{\tilde{\sigma}_t^k(\hat{\theta}_n^*)}\eta_t^{*k}-\mu_k\bigg)\\
%%%
=& \frac{1}{\sqrt{n}}\sum_{t=1}^n  \big(\eta_t^{*k}-\hat{\mu}_{n,k}\big)-\underbrace{\frac{1}{n}\sum_{t=1}^n  k\eta_t^{*k} \hat{D}_t'}_{I}\sqrt{n}\big(\hat{\theta}_n^*-\hat{\theta}_n\big)\\
%%%
&+\frac{1}{2}\sqrt{n}\big(\hat{\theta}_n^*-\hat{\theta}_n\big)'\underbrace{\frac{1}{n \sqrt{n}}\sum_{t=1}^n  \frac{\sigma_t^k(\hat{\theta}_n)}{\sigma_t^k(\breve{\theta}_n)}\Big(k(k+1)\tilde{D}_t(\breve{\theta}_n)\tilde{D}_t'(\breve{\theta}_n)-k\tilde{H}_t(\breve{\theta}_n)\Big)\eta_t^{*k}}_{II} \sqrt{n}\big(\hat{\theta}_n^*-\hat{\theta}_n\big)
\end{align*}
with $\breve{\theta}_n$ between $\hat{\theta}_n^*$ and $\hat{\theta}_n$. We find $I \overset{p^*}{\to}k\mu_k\Omega'$ almost surely since 
\begin{align*}
\EE^*[I]=&k \EE^*\big[\eta_t^{*k}\big]\frac{1}{n}\sum_{t=1}^n\hat{D}_t'\overset{a.s.}{\to}k\mu_k\Omega'\\
\Var^*[I]=&\frac{k^2}{n}\Var^*\big[\eta_t^{*k}\big]\frac{1}{n}\sum_{t=1}^n\hat{D}_t\hat{D}_t'\overset{a.s.}{\to}0\Var\big[\eta_t^k\big]J=O_{r\times r},
\end{align*}
where the convergence follows from \citeauthor{beutner2018residual} (\citeyear{beutner2018residual}, Lemma 2).  Consider the second term; since $\hat{\theta}_n\overset{a.s.}{\to}\theta_0$ (Theorem 1) and $\hat{\theta}_n^*\overset{p^*}{\to}\theta_0$  almost surely (\citeauthor{beutner2018residual}, \citeyear{beutner2018residual}, Lemma 5), we have $\PP^*\big[\breve{\theta}_n \notin \mathscr{V}(\theta_0)\big]\overset{a.s.}{\to}0$. Thus, for every $\varepsilon>0$ we obtain 
\begin{align*}
&\PP^*\big[||II||\geq \varepsilon \big]\\
%%%
\leq & \PP^*\Bigg[\bigg|\bigg|\frac{1}{n \sqrt{n}}\sum_{t=1}^n  \frac{\sigma_t^k(\hat{\theta}_n)}{\sigma_t^k(\breve{\theta}_n)}\Big(k(k+1)\tilde{D}_t(\breve{\theta}_n)\tilde{D}_t'(\breve{\theta}_n)-k\tilde{H}_t(\breve{\theta}_n)\Big)\eta_t^{*k}\bigg|\bigg|\geq \varepsilon \cap \breve{\theta}_n \in \mathscr{V}(\theta_0)\Bigg]\\
&\qquad + \PP^*\Big[\breve{\theta}_n \notin \mathscr{V}(\theta_0)\Big]\\
%%%
\leq & \PP^*\Bigg[\frac{1}{n \sqrt{n}}\sum_{t=1}^n\sup_{\theta \in \mathscr{V}(\theta_0)}\frac{\tilde{\sigma}_t(\hat{\theta}_n)}{\tilde{\sigma}_t(\theta)}k\bigg((k+1)\sup_{\theta \in \mathscr{V}(\theta_0)}\big|\big|\tilde{D}_t(\theta)\big|\big|^2+\sup_{\theta \in \mathscr{V}(\theta_0)}\big|\big|\tilde{H}_t(\theta)\big|\big|\bigg) \big|\eta_t^{*}\big|^k\geq \varepsilon\Bigg]+o(1)\\
%%%
\leq & \frac{1}{\varepsilon}\EE^*\Bigg[\frac{1}{n \sqrt{n}}\sum_{t=1}^n\sup_{\theta \in \mathscr{V}(\theta_0)}\frac{\tilde{\sigma}_t(\hat{\theta}_n)}{\tilde{\sigma}_t(\theta)}k\bigg((k+1)\sup_{\theta \in \mathscr{V}(\theta_0)}\big|\big|\tilde{D}_t(\theta)\big|\big|^2+\sup_{\theta \in \mathscr{V}(\theta_0)}\big|\big|\tilde{H}_t(\theta)\big|\big|\bigg) \big|\eta_t^{*}\big|^k\Bigg]+o(1)\\
%%%
= &\frac{1}{\varepsilon} \EE^*\big[|\eta_t^{*}|^k\big]\frac{1}{n \sqrt{n}}\sum_{t=1}^n\sup_{\theta \in \mathscr{V}(\theta_0)}\frac{\tilde{\sigma}_t(\hat{\theta}_n)}{\tilde{\sigma}_t(\theta)}k\bigg((k+1)\sup_{\theta \in \mathscr{V}(\theta_0)}\big|\big|\tilde{D}_t(\theta)\big|\big|^2+\sup_{\theta \in \mathscr{V}(\theta_0)}\big|\big|\tilde{H}_t(\theta)\big|\big|\bigg) +o(1)
\end{align*}
almost surely, where the third inequality follows from Markov's inequality.
Analogously to (\citeauthor{beutner2018residual}, \citeyear{beutner2018residual}, Equation A.71) one can show that
\begin{align}
\label{eq:4.A.71}
\frac{1}{n}\sum_{t=1}^n\sup_{\theta \in \mathscr{V}(\theta_0)}\frac{\tilde{\sigma}_t(\hat{\theta}_n)}{\tilde{\sigma}_t(\theta)}k\bigg((k+1)\sup_{\theta \in \mathscr{V}(\theta_0)}\big|\big|\tilde{D}_t(\theta)\big|\big|^2+\sup_{\theta \in \mathscr{V}(\theta_0)}\big|\big|\tilde{H}_t(\theta)\big|\big|\bigg)
\end{align}
is stochastically bounded. Together with $\EE^*\big[|\eta_t^{*}|^k\big]\leq \EE^*\big[\eta_t^{*2k}\big]^{\frac{1}{2}}\overset{a.s.}{\to}\EE\big[\eta_t^{2k}\big]^{\frac{1}{2}}<\infty$ we establish $||II||\overset{p^*}{\to}0$ in probability. Combining results, we have
\begin{align*}
\sqrt{n}\big(\hat{\mu}_{n}^*-\hat{\mu}_{n}\big) = \frac{1}{\sqrt{n}}\sum_{t=1}^n\big(h(\eta_t^*)-\hat{\mu}_{n}\big)-\nu \Omega'\sqrt{n}\big(\hat{\theta}_n-\theta_0\big)+o_{p^*}(1)
\end{align*}
in probability. Together with the expansion of \citeauthor{beutner2018residual} (\citeyear{beutner2018residual}, Equation 4.4):
\begin{align*}
\sqrt{n}\big(\hat{\theta}_n^*-\hat{\theta}_n\big) = \frac{1}{2}J^{-1} \frac{1}{\sqrt{n}}\sum_{t=1}^n \hat{D}_t\big(\eta_t^{2}-1\big)+o_{p^*}(1)
\end{align*}
almost surely, we establish
\begin{align*}
      \begin{pmatrix}
      \sqrt{n}(\hat{\theta}_n^*-\hat{\theta}_n)\\
    \sqrt{n}(\hat{\mu}_{n}^* - \hat{\mu}_{n})
      \end{pmatrix}
      = 
      \begin{pmatrix}
      \frac{1}{2}J^{-1}&O_{m\times m}\\
    -\frac{1}{2}\nu\Omega'J^{-1} & I_{m\times m}
      \end{pmatrix}
        \begin{pmatrix}
      \frac{1}{\sqrt{n}}\sum_{t=1}^n \hat{D}_t\big(\eta_t^{*2}-1\big)\\
\frac{1}{\sqrt{n}}\sum_{t=1}^n\big(h(\eta_t^*)-\hat{\mu}_n\big)
      \end{pmatrix}+o_{p^*}(1)
\end{align*}
in probability. Employing Lemma \ref{lem:7.2} completes the proof.
\qed
%\end{proof}
\vspace{0.5cm}

%\begin{proof}
\noindent
\textit{Proof of Proposition \ref{prop:7.1}.}
The claim is proven by induction. For $m=1$ we can decompose $A(\theta,\eta)=B_{0,1}(\theta)+B_{1,1}(\theta)\eta^2$, where $B_{0,1}(\theta)$ and $B_{1,1}(\theta)$ are matrices of dimension $(p+q)\times(p+q)$. Presuming $A(\theta,\eta)^{\otimes m}= \sum_{k=0}^m B_{k,m}(\theta) \eta^{2k}$ holds true, we obtain
\begin{align*}
&A(\theta,\eta)^{\otimes (m+1)}= A(\theta,\eta)^{\otimes m} \otimes A(\theta,\eta)=\bigg(\sum_{k=0}^m B_{k,m}(\theta) \eta^{2k}\bigg) \otimes \big(B_{0,1}(\theta)+B_{1,1}(\theta)\eta^2\big)\\
=& \underbrace{\big(B_{0,m}(\theta)\otimes B_{0,1}(\theta)\big)}_{=B_{0,m+1}(\theta)}+ \sum_{k=0}^{m-1} \underbrace{\Big(\big(B_{k+1,m}(\theta)\otimes B_{0,1}(\theta)\big)+\big(B_{k,m}(\theta) \otimes B_{1,1}(\theta)\big)\Big)}_{=B_{k+1,m+1}(\theta)}\eta^{2(k+1)}\\
&\quad +\underbrace{\big(B_{m,m}(\theta)\otimes B_{1,1}(\theta)\big)}_{=B_{m+1,m+1}(\theta)}\eta^{2(m+1)},
\end{align*}
which completes the induction step and verifies the lemma's claim.
\qed
%\end{proof}
% \begin{align}
% B_{01}(\theta) = \left(\begin{array}{@{}ccc|ccc@{}}
%     0 & \dots & 0 & 0 & \dots & 0 \\
%      & I_{(q-1)\times(q-1)} & O_{(q-1)\times 1} &  & O_{(q-1)\times p} & \\\hline
%     \alpha_1  & \dots & \alpha_q  & \beta_1  & \dots & \beta_p  \\
%      & O_{(p-1)\times q} &  &  & I_{(p-1)\times(p-1)} & O_{(p-1)\times 1}
%   \end{array}\right)
% \end{align}
% %
% and
% %
% \begin{align}
% B_{1,1}(\theta) = \left(\begin{array}{@{}ccc|ccc@{}}
%     \alpha_1 & \dots & \alpha_q  & \beta_1  & \dots & \beta_p  \\
%      & O_{(q-1)\times(q-1)} & O_{(q-1)\times 1} &  & O_{(q-1)\times p} & \\\hline
%     0  & \dots & 0  & 0  & \dots & 0  \\
%      & O_{(p-1)\times q} &  &  & O_{(p-1)\times(p-1)} & O_{(p-1)\times 1}
%   \end{array}\right)\\
%   = \left(\begin{array}{@{}ccc|ccc@{}}
%     \alpha' & \beta'  \\
%      O_{(p+q-1)\times q}& O_{(p+q-1)\times p} 
%   \end{array}\right)
% \end{align}

%\begin{proof}
%\citeauthor{francq2015risk} (\citeyear{francq2015risk}, Theorem 1) establish $\hat{\theta}_n \overset{a.s.}{\to} \theta_0$. The second claim follows from \citeauthor{beutner2018residual} (\citeyear{beutner2018residual}, Lemma 2).
%\end{proof}

\begin{lemma}
\label{lem:7.1}
Suppose Assumptions \ref{as:7.1}--\ref{as:7.3}, \ref{as:7.4}(\ref{as:7.4.1}), \ref{as:7.5}(\ref{as:7.5.1}), \ref{as:7.9}(i) hold with $a=d$. If, in addition $\tau(\theta_0,\mu)\leq 1$ and $||\mu||<\infty$ hold, then the estimator in \eqref{eq:238748761} is strongly consistent, i.e.\ $\hat{\theta}_n^c\overset{a.s.}{\to}\theta_0$.
\end{lemma}

\begin{proof}
Set $\hat{\theta}_n(z)=\arg\max_{\theta \in \Theta_z} \frac{1}{n}\sum_{t=1}^n \tilde{\ell}_t(\theta)$ with $\Theta_z = \big\{\theta \in \Theta : \tau(\theta,z)\leq 1\big\}$ and define $\theta_0(z) = \arg\max_{\theta \in \Theta_z}\EE[\ell_t(\theta)]$. Given $z$, $\Theta_z$ is a compact set and it follows by \citeauthor{francq2015risk} (\citeyear{francq2015risk}, Theorem 1) that $\hat{\theta}_n(z)\overset{a.s.}{\to}\theta_0(z)$. Because $\hat{\mu}_n\overset{a.s.}{\to}\mu$ (\citeauthor{beutner2018residual}, \citeyear{beutner2018residual}, Lemma 2), $\tau$ is continuous in both arguments and $\Theta = \Theta_\mu$ as $\tau(\theta_0,\mu)=1$, it follows that $\hat{\theta}_n^c = \hat{\theta}_n\big(\hat{\mu}_n\big)\overset{a.s.}{\to}\theta_0$.
\end{proof}

%\begin{proof}
\noindent
\textit{Proof of Corollary \ref{cor:7.2}.}
We define $M=M(\theta_0,\mu)=C(\theta_0,\mu)'C(\theta_0,\mu)$ with  $C(\theta_0,\mu)  = \sum_{k=0}^m B_{k,m}(\theta_0) \mu_{2k}$ and similarly set $\hat{M}_n = M(\hat{\theta}_n,\hat{\mu}_n)$, $\hat{M}_n^c = M(\hat{\theta}_n^c,\hat{\mu}_n)$ and $\hat{M}_n^\star = M(\hat{\theta}_n^\star,\hat{\mu}_n^\star)$. It is plain to see that all entries of $B_{0,1}(\cdot)$ and $B_{1,1}(\cdot)$ are differentiable at $\theta \in \mathring{\Theta}$, which carries over to $B_{k,m}(\cdot)$ for $k=0,1,\dots,m$ by the recursive structure. It follows that all entries of $C(\cdot,\cdot)$ and thus $M(\cdot,\cdot)$ are differentiable in both arguments. Therefore the gradient of $vech\big(M(\cdot,\cdot)\big)$ exists. Evaluating the gradient at $(\theta_0',\mu')'$, it is evident that all entries are non-negative and some are strictly positive since some entries are (sums and products of) polynomials of $\alpha_{0_i}$, $i=1,\dots,q$, $\beta_{0_j}$, $j=1,\dots,p$ and $\mu_{2k}$, $k=1,\dots,m$, which satisfy $\alpha_{0_i}>0$ and $\beta_{0_j}>0$ (see Assumption \ref{as:7.6}). The delta-method in conjunction with Theorem \ref{thm:7.2} implies that $\sqrt{n}\big(vech(\hat{M}_n)-vech(M)\big)$ converges in distribution to a Gaussian vector and hence $\sqrt{n}\big(\hat{M}_n-M\big) \overset{d}{\to} G$, where $G$ is a symmetric matrix with Gaussian entries. Holding ``the key to Kato'', we follow \citeauthor{watson1983statistics} (\citeyear{watson1983statistics}, Appendix B) leading to $\sqrt{n}(\hat{T}_n^2-T^2) \overset{d}{\to}\frac{tr\{G P_{\max}\}}{tr\{P_{\max}\}}$, where $P_{\max}$ is the projector associated with the largest eigenvalue of $M$ (see also \citeauthor{kato2013perturbation}, \citeyear{kato2013perturbation}). Note that the limiting random variable has a normal distribution. Applying the delta method once more, we obtain $\sqrt{n}(\hat{T}_n-T) \overset{d}{\to}\frac{1}{2T}\frac{tr\{G P_{\max}\}}{tr\{P_{\max}\}}$.

Consider the bootstrap and suppose that $\bar{H}_0:T \leq 1$ holds true. Repeating the previous argument in the bootstrap case while noting that $\lambda_{\max}(\hat{M}_n^c)=\hat{T}_n^{c\:2}\overset{a.s.}{\to}T^2$ and that the associated projector $\hat{P}_{n,\max}^c$ satisfies $\hat{P}_{n,\max}^c\overset{a.s.}{\to}P_{\max}$ as $\hat{M}_n^c=M(\hat{\theta}_n^c,\hat{\mu}_n) \overset{a.s.}{\to}M(\theta_0,\mu_0)=M$ (see \citeauthor{watson1983statistics}, \citeyear{watson1983statistics}), we establish $\sqrt{n}(\hat{T}_n^\star-\hat{T}_n^c)\overset{d^\star}{\to}\frac{1}{2T}\frac{tr\{G P_{\max}\}}{tr\{P_{\max}\}}$ in probability. Applying P\'olya's lemma (cf.\ \citeauthor{roussas1997course}, \citeyear{roussas1997course}, p.\ 206) validates \eqref{eq:5850320}.

We only highlight the changes when deriving the limiting distribution of the bootstrap quantity under the alternative $\bar{H}_1$. When $T=\tau(\theta_0,\mu)>1$, then $\hat{\theta}_n^c$ converges to a pseudo-true value, say $\theta^c$, such that $\hat{M}_n^{c}$ converges to $M^c=M(\theta^c,\mu)$. In that case $\sqrt{n}\big(\hat{M}_n^\star-\hat{M}_n^c\big) \overset{d^\star}{\to} G^c$ in probability, where $G^c$ is again a symmetric matrix with Gaussian entries. Furthermore,  $\lambda_{\max}(\hat{M}_n^c)=\hat{T}_n^{c\:2}$ converges to $T^c = \tau(\theta^c,\mu)$ whereas $\hat{P}_{n,\max}^c$ approaches $P_{\max}^c$, the projector associated with the largest eigenvalue of $M^c$. It follows that $\sqrt{n}(\hat{T}_n^\star-\hat{T}_n^c) \overset{d^\star}{\to}\frac{1}{2T^c}\frac{tr\{G^c P_{\max}^c\}}{tr\{P_{\max}^c\}}$ in probability, which establishes $\sqrt{n}(\hat{T}_n^\star-\hat{T}_n^c)=O_{p^\star}(1)$ in probability.
\qed
%\end{proof}
%Note that $T^2 = \lambda_{\max}(M)=1$ in this case  and the limiting random variable has a normal distribution. Applying the delta method once more, we obtain $\sqrt{n}(\hat{T}_n-1) \overset{d}{\to}\frac{1}{2}\frac{tr\{G P_{\max}\}}{tr\{P_{\max}\}}$. Repeating the argument in the bootstrap case while noting that $\lambda_{\max}(\hat{M}_n^c)=\hat{T}_n^{c\:2}=1$ and that the associated projector $\hat{P}_{n,\max}$ satisfies $\hat{P}_{n,\max}\overset{a.s.}{\to}P_{\max}$ as $\hat{M}_n^c \overset{a.s.}{\to}M$ (see again \citeauthor{watson1983statistics}, \citeyear{watson1983statistics}), we establish $\sqrt{n}(\hat{T}_n^\star-1)\overset{d^\star}{\to}\frac{1}{2}\frac{tr\{G P_{\max}\}}{tr\{P_{\max}\}}$ in probability.

\section*{Acknowledgements}
The paper was drafted during the author's research visit at CREST, Paris. The author thanks  Jean-Michel Zako\"ian, Christian Francq, Cees Diks, Eric Beutner, Stephan Smeekes, Sean Telg and Hanno Reuvers for useful comments and suggestions. This research was financially supported by the Netherlands Organisation for Scientific Research (NWO).

\singlespacing
\bibliographystyle{Chicago-modified}
\bibliography{bibliography}

\end{document}